\documentclass[13pt,a4paper,fleqn]{article}
\usepackage[T1]{fontenc}
\usepackage{jheppub}
\usepackage{alemfonts}
\usepackage[utf8]{inputenc}
\usepackage[english]{babel}
\usepackage{amsmath}
\usepackage{amsfonts}
\usepackage{amssymb}
\usepackage{xcolor}
\usepackage{mathrsfs}
\usepackage{amsthm}
\usepackage{float}
\usepackage{graphicx}
\usepackage[linesnumbered,lined,boxed,commentsnumbered]{algorithm2e}
\usepackage[caption = false]{subfig}
\usepackage[parfill]{parskip}

\numberwithin{equation}{section}

\long\def \beq#1\eeq {\begin{equation} #1 \end{equation}}
\long\def \beaq#1\eeaq {\begin{equation}\begin{aligned} #1 \end{aligned}\end{equation}}
\long\def \bes#1\ees {\begin{equation}\begin{split} #1 \end{split} \end{equation}}
\long\def \bea#1\eea {\begin{eqnarray} #1 \end{eqnarray}}
\long\def \bse[#1]#2\ese {\begin{subequations}\label{#1}\begin{align} #2 \end{align}\end{subequations}}

\newcommand{\evalat}[2]{\left. #1 \right\vert_{#2}}

\newcommand{\mv}[1]{\langle #1\rangle}
\newcommand{\bt}{\tilde{t}}
\newcommand{\mvt}[1]{\langle #1\rangle_t}

\newcommand{\pder}[2]{\ensuremath{\frac{\partial #1}{\partial #2}}}
\newcommand{\fder}[2]{\ensuremath{\frac{\delta #1}{\delta #2}}}
\long\def\dm[#1]{\!\operatorname{d\mu}\left(#1\right)}

\newcommand{\RS}{{\rm \scriptscriptstyle RS}}
\theoremstyle{plain}
\newtheorem{Remark}{Remark}
\newtheorem{Theorem}{Theorem}
\newtheorem{Lemma}{Lemma}
\newtheorem{Proposition}{Proposition}
\newtheorem{Corollary}{Corollary}
\newtheorem{Definition}{Definition}

\graphicspath{{./Figures/}}

\title{Generalized Guerra's interpolation schemes for dense associative neural networks}
\author[a]{Elena Agliari}
\author[b,c]{Francesco Alemanno}
\author[b,d]{Adriano Barra}
\author[b,d]{Alberto Fachechi}

\affiliation[a]{Dipartimento di Matematica  {\em Guido Castelnuovo}, Italy}
\affiliation[b]{Dipartimento di Matematica e Fisica {\em Ennio De Giorgi}, Universit\`a del Salento, Italy}
\affiliation[c]{C.N.R. Nanotec Lecce, Italy}
\affiliation[d]{I.N.F.N., Sezione di Lecce, Italy}

\abstract{In this work we develop analytical techniques to investigate a broad class of associative neural networks set in the high-storage regime. These techniques translate the original statistical-mechanical problem into an analytical-mechanical one which implies solving a set of partial differential equations, rather than tackling the canonical probabilistic route. We test the method on the classical Hopfield model -- where the cost function includes only two-body interactions (i.e., quadratic terms)
 -- and on the ``relativistic'' Hopfield model -- where the (expansion of the) cost function includes $p$-body (i.e., of degree $p$) contributions. 
Under the replica symmetric assumption, we paint the phase diagrams  of these models by obtaining the explicit expression of their free energy as a function of the model parameters (i.e., noise level and memory storage).  Further, since for non-pairwise models ergodicity breaking is non necessarily a critical phenomenon, we develop a fluctuation analysis and find that criticality is preserved in the relativistic model.}

\begin{document}

\maketitle



\section{Introduction}
In the last few years, the introduction of a new-generation of computational devices \cite{GPU1,GPU2}, the development of better performing learning schemes \cite{Hinton1,DL1} and the implementation of massive data repositories allowing efficient training \cite{Clouds1,Cloud2} naturally raised the quest for appropriate mathematical tools and frameworks able to describe, address and possibly explain the related information processing capabilities \cite{DL1}.
\newline
Since the 80s', the statistical mechanics of spin glasses \cite{MPV} has been playing a primary role in the investigation of neural networks, as for both their learning skills \cite{angel-learning,sompo-learning} and their retrieval properties \cite{Amit,Coolen}. Along the past decades, beyond the bulk of results achieved via heuristic approaches like the replica trick \cite{Amit,MPV}, a considerable amount of rigorous results exploiting alternative routes was also developed (see e.g. \cite{Agliari-Barattolo,ABT,BovierBook,Bovier1,Bovier2,Bovier3,Albert2,Barra-JSP2010,bipartiti,Dotsenko1,Dotsenko2,Tala1,Tala2,Tirozzi,Pastur}).

This paper goes in the last direction and aims bridging between a statistical-mechanics approach and a mechanical approach, the latter possibly more familiar to a wider community.
In particular, we shall focus on the Hopfield model (as a reference framework) and its ``relativistic'' generalization \cite{Albert1}. The latter, exhibiting a cost function that is an (infinite and convergent) series of monomials in the microscopic variables (i.e., the neural activities) offers not only a perfect ``playground'' where testing our methods, but also an interesting example of dense architectures \cite{Krotov1,Krotov2}.
The translation of the statistical-mechanics problem into a mechanical framework is based on the analogy between the variational principles in statistical mechanics (e.g., maximum entropy and minimal energy) and the least-action principle in analytical mechanics: this route was already paved for ferromagnetic models \cite{Bogo1,Bogo2,Moro}, for spin glasses \cite{Barra0,BarraHJ,GuerraSum} and for (simpler) neural networks \cite{Agliari-Barattolo,Albert1}. A main advantage is that it allows painting the phase diagrams of the model under study by relying upon tools originally developed in the analytical counterpart (i.e. mainly differential equation techniques) almost without any knowledge of the statistical mechanics of complex systems.

We stress that in this work we deal with models in the high-storage regime, namely, we set the number $P$ of stored patterns to be retrieved scaling linearly with the number $N$ of neurons making up the network, i.e., the load $\alpha := \lim_{N \to \infty}P/N$ is non vanishing. This regime is notoriously difficult to handle rigorously, given the emergence of glassy phenomena impairing the retrieval capabilities of the system (see e.g., \cite{Amit,Coolen}). In fact, the ``relativistic'' Hopfield model has already been addressed in the low-storage regime, namely when $\alpha=0$, while
when $\alpha >0$ the control of the model was still unsatisfactory \cite{Albert1}. Here we overcome that limit by generalizing Guerra's interpolation technique in order to be able to work with a broad class of Hamiltonians, that are expressable as a smooth function $\mathcal F$ of the squared Mattis magnetization $\boldsymbol m$, that is, $H =  \mathcal{F}(\boldsymbol m^2) = \mathcal{F}(\sum_{\mu} m_{\mu}^2)$, where $m_{\mu}$ measures the retrieval of the $\mu$-th stored pattern ($\mu=1,...,P$). 
Remarkably, this general model also includes a subset of ``dense'' networks, where ``dense'' refers to the presence of $p$-wise ($p>2$) interactions among neurons. 

This paper is structured as follows: as a preamble, in Sec.~\ref{riscaldamento}, we solve the classical Hopfield model in the high-storage regime via the mechanical analogy and, by solving a set of differential equations, we recover the standard Amit-Gutfreund-Sompolinsky (AGS) quenched free-energy and related self-consistencies for the order parameters \cite{Amit}; while in the original statistical mechanical setting this result emerges through a rather tricky and lengthy procedure, here it is just an almost trivial exercise. In Sec.~\ref{definizioni} the ``relativistic'' Hopfield model in the high-storage regime is introduced and embedded in its statistical mechanical framework. In Sec.~\ref{momentoprimo}, by relying upon a generalization of the mechanical analogy, we calculate explicitly its quenched free energy (at the replica symmetric level, as standard) and we provide a first picture of the phase diagram of the model. Next, in Sec.~\ref{momentosecondo}, we study the fluctuations of the order parameters to inspect ergodicity breaking and its possible critical nature. By combining results from Secs.~\ref{momentoprimo} and \ref{momentosecondo} a consistent picture of the phase diagram of the relativistic model is finally obtained. Then, Sec.~\ref{conclusions} is left for conclusions and speculations on the broad applicability of these new mathematical techniques. Finally, in Appendix A we present the solution of the ``relativistic'' Hopfield model via replica trick.

\section{The classical Hopfield network via the mechanical analogy}\label{riscaldamento}

The Hopfield model is the prototype model for neural networks performing pattern recognition; soon after the seminal paper by J.J. Hopfield \cite{Hopfield}, the statistical-mechanics analysis by Amit, Gutfreund and Sompolinsky \cite{AGS1,AGS2} highlighted a very rich phenomenology which prompted an upsurge of interest among physicists and mathematicians. Admittedly, the statistical-mechanics approach is strongly grounded on frustrated systems with quenched disorder (i.e., {\em spin-glasses}), typically studied just in some fields of Theoretical Physics and Applied Mathematics. Alternative investigation routes, possibly based on more widespread frameworks, may allow for a deeper comprehension and may facilitate interdisciplinarity with a new wave of ideas and inspirations.
Hereafter we solve the Hopfield model by exploiting differential equation theory, standard in many transport problems. \newline In the next definitions we introduce the fundamental quantities we will deal with.
\begin{Definition} Set $\alpha \in \mathbb{R}^+$  and let $\boldsymbol \sigma \in \{- 1, +1\}^N$ be a configuration of $N$ binary neurons. Given $P=\alpha N$ random patterns $\{\boldsymbol \xi^{\mu}\}_{\mu=1,...,P}$, each made of $N$ digital entries identically and independently drawn with probability $P(\xi_i^{\mu}=+1)=P(\xi_i^{\mu}=-1)=1/2$, for $i=1,...,N$, the classical Hopfield cost-function (or ``Hamiltonian'' to preserve a physical jargon) is
	\beq
	H_N(\boldsymbol \sigma| \boldsymbol \xi) := -\frac{1}{2N}\sum_{\mu=1}^{P}\sum_{i,j=1}^{N,N}\xi_i^{\mu}\xi_j^{\mu}\sigma_i\sigma_j.
	\eeq
	\label{hop_hbare}
\end{Definition}
\begin{Definition}The partition function related to the cost-function (\ref{hop_hbare}) is given by
	\beq
	Z_N(\beta,\boldsymbol \xi) := \sum_{\{ \boldsymbol \sigma \}} \exp \left[ -\beta H_N(\boldsymbol \sigma | \boldsymbol \xi) \right]=\sum_{\{ \boldsymbol \sigma \}} \exp \left( \frac{\beta}{2N}\sum_{\mu=1}^{P}\sum_{i,j=1}^{N,N}\xi_i^{\mu}\xi_j^{\mu}\sigma_i\sigma_j\right),
	\eeq
	where $\beta \geq 0$ is a real number accounting for the fast noise (or ``inverse temperature'') in the network and such that for $\beta \to 0$ the probability distribution for the neural configuration is uniformly spread while for $\beta \to \infty$ it is sharply peaked at the minima of the cost function (\ref{hop_hbare}).
	\label{hop_BareZ}
\end{Definition}
\begin{Definition} \label{def:Bolt}
	The Boltzmann average induced by the partition function (\ref{hop_BareZ}) is denoted with $\omega_{\boldsymbol \xi}$ and, for an arbitrary observable $O(\boldsymbol \sigma)$, is defined as
	\begin{equation}
	\omega_{\boldsymbol \xi} (O) : = \frac{\sum_{\boldsymbol \sigma} O(\boldsymbol \sigma) e^{- \beta H_N(\boldsymbol \sigma| \boldsymbol \xi)}}{Z_N(\beta, \boldsymbol \xi)}.
	\end{equation}
	This can be further averaged over the realization of the $\boldsymbol \xi$'s, also referred to as quenched average and denoted as $\mathbb E$, to get
	\beq
	\langle O \rangle := \mathbb{E} \omega_{\boldsymbol \xi} (O).
	\eeq
	Further, we introduce the product state $\Omega_{s, _{\boldsymbol \xi}} = \omega_{\boldsymbol \xi}^{(1)} \times \omega_{\boldsymbol \xi}^{(2)} \times ... \times \omega_{\boldsymbol \xi}^{(s)}$ over $s$ replicas of the system, characterized by the same realization $\boldsymbol \xi$ of disorder. In the following, we shall use the product state over two replicas only, hence we shall neglect the index $s$ without ambiguity; also, to lighten the notation, we shall also omit the subscript $\boldsymbol \xi$ in $\omega_{\boldsymbol \xi}$ and in $\Omega_{\boldsymbol \xi}$. Thus, for an arbitrary observable $O(\boldsymbol \sigma^{(1)}, \boldsymbol \sigma^{(2)})$
	\beq
	\langle O \rangle := \mathbb{E} \Omega (O) = \mathbb{E}  \frac{\sum_{\boldsymbol \sigma} O(\boldsymbol \sigma^{(1)},\boldsymbol \sigma^{(2)}) e^{-\beta [H_N(\boldsymbol \sigma^{(1)}| \boldsymbol \xi) + H_N(\boldsymbol \sigma^{(2)}| \boldsymbol \xi)] }}{Z_N^2(\beta, \boldsymbol \xi)},
	\eeq
	where $\boldsymbol{\sigma}^{(1,2)}$ is the configuration pertaining to the replica labelled as $1,2$.
\end{Definition}
\begin{Definition}The intensive quenched pressure of the classical Hopfield model (\ref{hop_hbare}) is defined as
	\beq
	A_N(\alpha, \beta) := \frac{1}{N} \bbE \log Z_N(\beta, \boldsymbol \xi),
	\label{hop_BareA}
	\eeq
	and its thermodynamic limit is referred to as
	\beq
	A(\alpha, \beta) := \lim_{N \to \infty} A_N(\alpha, \beta).
	\eeq
\end{Definition}
\begin{Remark}
	The quenched pressure $A_N(\alpha, \beta)$ is related to the quenched free-energy $F_N(\alpha, \beta)$ as $A_N(\alpha, \beta) = - \beta F_N(\alpha, \beta)$ and here it is chosen for mathematical convenience.
\end{Remark}

Before proceeding it is worth recalling the {\em universality property} of the quenched noise in spin-glasses: when dealing with mean-field spin-glasses \cite{Univ1} and mean-filed bipartite spin-glasses \cite{Genovese}, the coupling distribution is proved not to affect the resulting pressure, provided that it is centered, symmetrical and with finite variance. This property was extended to the Hopfield model in \cite{Agliari-Barattolo}, where it was shown that the quenched noise contribution appearing in the expression for the pressure (which stems from the $P-1$ non-retrieved patterns and which tends to inhibit retrieval), exhibits the very same shape, regardless of the nature of the pattern entries (that is, digital -- e.g., Boolean -- or analog -- e.g., Gaussian).

Focusing on pure retrieval states, we will assume without loss of generality that the candidate pattern to be retrieved $\boldsymbol \xi^1$ is a Boolean vector of $N$ entries, while $\boldsymbol \xi^{\mu}$, $\mu=2,...,P$ are real vectors whose $N$ entries are i.i.d. standard Gaussian. Accordingly, the average $\bbE$ acts as a Boolean average over $\boldsymbol \xi^1$ and as a Gaussian average over $\boldsymbol \xi^2 \cdots \boldsymbol \xi^P$.

\begin{Definition} The order parameters used to describe the macroscopic behavior of the model are the standard ones \cite{Amit,Coolen}, namely, the Mattis magnetization 
	\begin{equation}
	m(\boldsymbol \sigma) :=m(\boldsymbol \sigma| \boldsymbol \xi) := \frac{1}{N}\sum_{i=1}^{N} \xi^1_1 \sigma_i
	\end{equation}
	to quantify the retrieval capabilities of the network, and the two-replica overlap in the $\boldsymbol \sigma$'s variables
	\beq
	\label{q}
	q_{12}(\boldsymbol \sigma) := \frac{1}{N}\sum_{i=1}^N \sigma_i^{(1)}\sigma_i^{(2)}
	\eeq
	to quantify the level of slow noise the network must cope with when performing pattern recognition.
	Further, as an additional set of variables $\{\tau_{\mu}\}_{\mu=1,...,P}$ shall be introduced (vide infra), we accordingly define the related two-replica overlap
	\beq
	\label{p}
	p_{11}(\boldsymbol \tau) :=  \frac{1}{P} \sum_{\mu=1}^P \tau^{(1)}_\mu \tau^{(1)}_\mu, ~~~ p_{12}(\boldsymbol \tau) :=  \frac{1}{P} \sum_{\mu=1}^P \tau^{(1)}_\mu \tau^{(2)}_\mu
	\eeq
	which as well captures the level of the noise due to pattern inference.
	\label{hop_orderparameters}
\end{Definition}

\begin{Definition} Given five real interpolating parameters $(t, x, y, z, w)$ to be set a posteriori and $N+P$ auxiliary quenched i.i.d. random variables $J_i \sim \mathcal{N}[0,1], i \in (1,...,N)$ and  $\tilde{J}_{\mu} \sim \mathcal{N}[0,1], \mu \in (1,...,P)$, the interpolating pressure -- also known as {\em Guerra's action} -- for the classical Hopfield model (\ref{hop_hbare}) in the high-storage regime is defined as
	\bes
	\mathcal{A}_N(t,x,y,z,w) &:= \frac{1}{N} \bbE \log \sum_{\{ \boldsymbol \sigma \} } \int \calD\tau \exp \Big[\frac{\sqrt{t}}{\sqrt{N}} \sum_{i,\mu >1}^{N,P}\xi_i^{\mu}\sigma_i\tau_\mu +\frac{t N}{2}m^2(\boldsymbol \sigma) + \sqrt{x} \sum_{i=1}^N J_i \sigma_i +\\
	&
	+ \sqrt{y} \sum_{\mu=1}^P \tilde J_\mu \tau_\mu + z \sum_{\mu=1}^P \frac{\tau_\mu^2}{2}+w N m(\boldsymbol \sigma)\Big],
	\label{hop_GuerraAction}
	\ees
	where $\calD \tau := \prod_{\mu=1}^P \frac{e^{-\tau_{\mu}^2/2}}{\sqrt{2\pi}}$ 
	is the standard Gaussian measure over the $P$ real auxiliary variables $\tau_{\mu}, \ \ \mu \in (1,...,P)$, and the expectation $\mathbb E$ is now meant over $\boldsymbol \xi$, $\boldsymbol J$, and $\boldsymbol{\tilde{J}}$.
\end{Definition}
%
\begin{Remark}
	We interpret $\textbf{r} = (x,y,z,w)$ as the spatial coordinate in a four-dimensional Euclidean space and $t$ as the temporal variable. According to the mechanical analogy, we will consider the spatial coordinates as functions of $t$ (we will specify their form in the next theorem) and we require that $\textbf{r}(t=\beta) \overset{!}{=} 0$, in such a way that, when the temporal coordinate $t$ is equal to $\beta$, we recover the original pressure \eqref{hop_BareA}, namely $\mathcal A_N(t=\beta,\textbf{r} =0)=A_N(\alpha,\beta)$.
\end{Remark}
\begin{Definition}
	Since the variables $x,y,z,w$ are functions of $t$, we will refer to the Boltzmann average stemming from the interpolating system as
	$\omega_{t} (\cdot)$, that is,
	\begin{equation}
	\omega_t(\cdot)=\frac{\sum_{\boldsymbol \sigma}\int \calD \tau ~(\cdot)~ B_{t,\boldsymbol r}(\boldsymbol \sigma, \boldsymbol \tau)}{Z_{t,\boldsymbol r}},
	\end{equation}
	where $B_{t,\boldsymbol r}$ is the generalized Boltzmann factor
	\bea
	\nonumber
	B_{t,\boldsymbol r} (\boldsymbol \sigma, \boldsymbol \tau) :=  \exp \Big[\frac{\sqrt{t}}{\sqrt{N}} \sum_{i,\mu >1}^{N,P}\xi_i^{\mu}\sigma_i\tau_\mu +\frac{t N}{2}m^2(\boldsymbol \sigma) + \sqrt{x} \sum_i J_i \sigma_i + \sqrt{y} \sum_\mu \tilde J_\mu \tau_\mu + z \sum_\mu \frac{\tau_\mu^2}{2}+w N m(\boldsymbol \sigma)\Big],
	\eea
	and
	\begin{equation}
		Z_{t,\boldsymbol r}=\sum_{\boldsymbol \sigma}\int \calD \tau B_{t,\boldsymbol r}(\boldsymbol \sigma, \boldsymbol \tau).
	\end{equation}
	In analogy with the non-deformed average, we also pose
	\beq
		\mvt{\cdot} := \mathbb{E} \omega_{t} (\cdot),
		\eeq
		where now $\mathbb E$ means also the average over the $\boldsymbol{J}$ and $\boldsymbol {\tilde{J}}$ variables.
\end{Definition}
\begin{Lemma} The partial derivatives of the Guerra Action (\ref{hop_GuerraAction}) w.r.t. $t,x,y,z,w$ can be expressed in terms of the generalized expectations of the order parameters as
	\bea
	\label{hop_expvalsa}
	\frac{\partial \mathcal{A}_N}{\partial t} &=& \frac{\alpha}{2} \big[ \mvt{p_{11}}- \mvt{p_{12}q_{12}}\big] + \frac{1}{2} \mvt{m^2}\\
	\label{hop_expvalsmiddle}
	\frac{\partial \mathcal{A}_N}{\partial x}  &=& \frac{1}{2}  \big[1- \mvt{q_{12}}\big],\\
	\frac{\partial \mathcal{A}_N}{\partial y}  &=& \frac{\alpha}{2} \big[\mvt{p_{11}}-\mvt{p_{12}}\big],\\
	\frac{\partial \mathcal{A}_N}{\partial z}  &=& \frac{\alpha}{2}  \mvt{p_{11}},\\
	\frac{\partial \mathcal{A}_N}{\partial w}  &=& \mvt{m},
	\label{hop_expvalsb}
	\eea
	
	\begin{proof}
		The previous equations can be obtained by straightforward computations. We provide details for (\ref{hop_expvalsa}) only, the other derivations work analogously. Deriving $\mathcal{A}_N(t,\boldsymbol r)$ with respect to $t$ we get
		\bea
		\nonumber
		\frac{\partial \mathcal{A}_N}{\partial t} &=& \frac{1}{N} \mathbb{E} \frac{ \sum_{\boldsymbol \sigma} \int \calD\tau \left [ \frac{1}{2\sqrt{tN}} \sum_{i,\mu >1}^{N,P}\xi_i^{\mu}\sigma_i\tau_\mu +\frac{N}{2}m^2(\boldsymbol \sigma) \right ] B_{t,\boldsymbol r}(\boldsymbol \sigma, \boldsymbol \tau)  }{Z_{t,\boldsymbol r}}= \\
		\label{pp}
		&=& \frac{1}{2 N \sqrt{tN} } \mathbb{E} \frac{ \sum_{\boldsymbol \sigma} \int \calD\tau  \sum_{i,\mu >1}^{N,P}\xi_i^{\mu}\sigma_i\tau_\mu  B_{t,\boldsymbol r} (\boldsymbol \sigma, \boldsymbol \tau) }{Z_{t,\boldsymbol r}} + \frac{1}{2} \mvt{m^2}.
		\eea
		Now, we still need to handle the first term in the right-hand-side of eq.~(\ref{pp}); by applying Wick's theorem $\mathbb{E} [\xi_i^{\mu} f(\boldsymbol \xi)] = \mathbb{E} \partial_{\xi_i^{\mu}} f(\boldsymbol \xi)$ to that term, we get
		\bea
		\nonumber
		\frac{\partial \mathcal{A}_N}{\partial t} &=&  \frac{1}{2 N ^2 } \sum_{i,\mu >1}^{N,P}\big(\mvt{\tau_\mu^2}-\mvt{\sigma_i \tau_\mu}^2\big) + \frac{1}{2} \mvt{m^2}=\\
		&=&\frac\alpha2\big(\mvt{p_{11}}-\mvt{p_{12}q_{12}}\big)+ \frac{1}{2} \mvt{m^2}
		\eea		
		where we used the definitions \eqref{q} and \eqref{p} and replaced $P-1$ with $\alpha N$ (which is valid in the thermodynamic limit). 
	\end{proof}
\end{Lemma}
\begin{Definition}
	Under the replica-symmetry (RS) assumption, in the thermodynamic limit, the distribution of the generic order parameter $X$ is centered at its expectation value $\bar X(t)$ with vanishing fluctuations for all $t$, that is, it converges to a Dirac delta: $\lim _{N\to \infty} P_t(X)=\delta(X-\bar X(t))$. Otherwise stated, being
		\beq \label{eq:delta}
	\Delta X :=  X - \bar{X}(t),
	\eeq
	for all $t \in \bbR^+$ we have $\langle \Delta X(t) \rangle_t \overset{N\to\infty}{\longrightarrow}0$ and $\mvt{\Delta X \Delta Y} = 0$ for any generic pair of order parameters $X,Y$. We therefore define the following expectation values
	\bea
	\bar m(t) &=&\lim_{N\to \infty} \mv{m}_t,\\
	\bar q(t) &=&\lim_{N\to \infty}  \mv{q_{12}}_t, \\
	\bar p(t) &=&\lim_{N\to \infty} \mv{p_{12}}_t.
	\eea
\end{Definition}
%
\begin{Proposition} The Guerra Action \eqref{hop_GuerraAction} obeys the following differential equation:
	{\small\beq
	\pder{\mathcal{A}_N}{t}-\alpha \bar p(t) \pder{\mathcal{A}_N}{x}-\bar q(t) \pder{\mathcal{A}_N}{y}-\Big(1-\bar q(t)\Big)\pder{\mathcal{A}_N}{z} -\bar m(t) \pder{\mathcal{A}_N}{w}=-\frac{\alpha}{2}\bar p(t)\Big(1-\bar q(t)\Big) - \frac{1}{2}\bar m^2(t)+V_N(t),
	\label{hop_GuerraAction_DE}
	\eeq
		}
	where
	\beq
	V_N(t)=\frac{1}{2}\mvt{(\Delta m)^2}-\frac{1}{2}\mvt{\Delta p_{12}\Delta q_{12}},
	\eeq
and, in the thermodynamic limit, under the RS assumption, one has
	\beq
	\pder{\calA}{t}-\alpha \bar p(t) \pder{\calA}{x}-\bar q(t) \pder{\calA}{y}-\Big(1-\bar q(t)\Big)\pder{\calA}{z} -\bar m(t) \pder{\calA}{w}=-\frac{\alpha}{2}\bar p(t)\Big(1-\bar q(t)\Big) - \frac{1}{2}\bar m^2(t).
	\label{hop_GuerraAction_RSDE}
	\eeq
	\begin{proof}
		First, we write the generalized averages $\langle \cdot \rangle_t$ appearing in \eqref{hop_expvalsa} in terms of the $\Delta$ operator (\ref{eq:delta}) as
		\bea
		\mvt{m^2} &=& -\bar m^2(t) +2 \bar m \mvt{m} + \mvt{(\Delta  m)^2},\\
		\mvt{p_{12}q_{12}} &=& -\bar p (t) \bar q(t) +\bar p(t) \mvt{q_{12}}+\bar q(t) \mvt{p_{12}} + \mvt{\Delta p_{12}\Delta q_{12}},
		\eea
		and, by plugging these espressions into \eqref{hop_expvalsa}, we get
		\bes
		\pder{A_N}{t}=\frac{\alpha}{2}\mvt{p_{11}}-\frac{\alpha}{2}\Big(-\bar p (t) \bar q(t) +\bar p(t) \mvt{q_{12}}+\bar q(t) \mvt{p_{12}}\Big) + \frac{1}{2} \Big(\bar m^2(t) +2 \bar m \mvt{m} \Big)+V_N(t).
		\ees
		Then, the generalized averages appearing in the previous expression can be recast in terms of the spatial derivatives of the Guerra Action $\mathcal A _N$ using the relations \eqref{hop_expvalsmiddle}-\eqref{hop_expvalsb}; after some trivial computations we get \eqref{hop_GuerraAction_DE}. \par
		Now, under the RS assumption, the averages $\mvt{\Delta X}$ and $\mvt{\Delta X \Delta Y}$ for any generic pair of order parameters $X,Y$ and all $t \in \bbR^+$ vanish in the thermodynamic limit. As a consequence, we have $\lim _{N\to \infty }V_N(t)=0$, and \eqref{hop_GuerraAction_DE} reduces to \eqref{hop_GuerraAction_RSDE}.
	\end{proof}
\end{Proposition}


\begin{Proposition}
	The unique solution to the differential equation \eqref{hop_GuerraAction_RSDE} is given by the functional
	\bes
	\calA^{(t,\boldsymbol r)}_{[\bar m,\bar p,\bar q]}=&\log 2 + \bbE_J \log\cosh\Big[w+ \int_0^t \bar m(\tilde t)\,\textrm{d} \tilde t +J\sqrt{x+\alpha \int_0^t \bar p (\tilde t)\,\textrm{d} \tilde t}\Big]+\\
	&+\frac{\alpha}{2}\frac{y+\int_0^t \bar q(\tilde t) \, \textrm{d} \tilde t}{1-z-\int_0^t [1-\bar q( \tilde t)]\,\textrm{d} \tilde t} -\frac{\alpha}{2}\log\Big(1-z-\int_0^t [1-\bar q(\tilde t)]\,\textrm{d} \tilde t\Big)+\\
	&-\frac{1}{2}\int_0^t \textrm{d} \tilde t\,[\alpha	\bar p(\tilde t)\big(1-\bar q(\tilde t)\big) +\bar m^2(\tilde t) ],
	\label{hop_mechanicalsolution}
	\ees
	which is a function of the parameters $t,x,y,z,w$ and a functional of the expectactions of the order parameters $(\bar m,\bar p,\bar q)(t)$.
	\begin{proof}
		The differential equation in \eqref{hop_GuerraAction_RSDE} can be solved with different methods. The simplest is perhaps the method of characteristics, whose key step is to find the ``characteristic curves'' of the PDE along which the PDE turns into an ODE. In this case, the characteristic curves are expressed in terms of the expectaction values of the order parameters as
		\bea
		x(t) &=& x_0 -\alpha\int_{t_0}^t \bar p (\bt) d\bt,\label{eq:char1}\\
		y(t) &=& y_0 -\int_{t_0}^t \bar q (\bt) d\bt,\\
		z(t) &=& z_0 -\int_{t_0}^t [1-\bar q (\bt)] d\bt,\\
		w(t) &=& w_0 -\int_{t_0}^t \bar m (\bt) d\bt\label{eq:charlast},
		\eea
		where $(x_0, y_0, z_0, w_0) = \boldsymbol r (t=t_0) = \boldsymbol{r_0}$.
		We can easily verify that, along these curves, $\frac{\rm d}{\textrm{d} t}\calA (t,\boldsymbol r (t))$ gives us exactly the l.h.s. of \eqref{hop_GuerraAction_RSDE}. Therefore, we can write
		\bea
		\frac{\rm d}{\textrm{d} t}\calA (t,\boldsymbol r (t))=-\frac{\alpha}{2}\bar p(t)\Big(1-\bar q(t)\Big) - \frac{1}{2}\bar m^2(t).
		\eea
		By integrating between $t_0$ and $t$ the previous equation we get
		\beq \label{eq:238}
		\calA (t,\boldsymbol r (t))=\calA (t_0,\boldsymbol r_0)-\frac{1}{2}\int_{t_0}^t d\bt \Big[\alpha\bar p(\tilde t)\Big(1-\bar q( \tilde t)\Big) + \bar m^2(\tilde t)\Big].
		\eeq
		Here, the parameter $t_0$ is free and can be chosen in order to simplify the computations. By inspecting \eqref{hop_GuerraAction}, we see that the most convenient choice is $t_0 = 0$, leaving us with the evaluation of a 1-body problem which can be directly solved:
		\bes \label{eq:239}
		\calA(t_0=0,\boldsymbol r _0)=&\log 2 + \bbE_J \log\cosh\Big[w_0 +J\sqrt{x_0}\Big]+\frac{\alpha}{2}\frac{y_0}{1-z_0} -\frac{\alpha}{2}\log\Big(1-z_0\Big).
		\ees
		Now, from (\ref{eq:char1})-(\ref{eq:charlast}) we get $x_0 = x(t) + \alpha \int_0^t \bar{p}(\tilde t) d \tilde t$, and similarly for $y_0, w_0, z_0$. These expressions can then be plugged in (\ref{eq:239}) and the resulting formula can be used in (\ref{eq:238}), finally obtaining (\ref{hop_mechanicalsolution}).
	\end{proof}
\end{Proposition}
\begin{Theorem}The replica-symmetry quenched pressure of the Hopfield model \eqref{hop_BareA} in the infinite volume limit is given by
	\bes
	A_{\rm RS}(\alpha,\beta)=&\log 2 + \bbE_J \log\cosh\Big[\int_0^\beta \bar m(t)\,\textrm{d}t +J\sqrt{\alpha \int_0^\beta \bar p (t)\,\textrm{d}t}\Big]+\frac{\alpha}{2}\frac{\int_0^\beta \bar q(t) \, \textrm{d}t}{1-\int_0^\beta [1-\bar q(t)]\,\textrm{d}t} \\
	-& \frac{\alpha}{2}\log\Big(1-\int_0^\beta [1-\bar q(t)]\,\textrm{d}t\Big)-\frac{1}{2}\int_0^\beta \textrm{d}t\,[\alpha	\bar p(t)\big(1-\bar q(t)\big) +\bar m^2(t) ].
	\label{hop_solution}
	\ees
	The self-consistency equations obtained from the quenched pressure \eqref{hop_solution} for $t' \in [0, \beta]$ are:
	\bes
	\bar m(t') &=\bbE_J \tanh\Big[\int_0^\beta \bar m(t)\,\textrm{d}t +J\sqrt{\alpha \int_0^\beta \bar p (t)\,\textrm{d}t}\Big],\\
	\bar q(t') &=\bbE_J \tanh^2\Big[\int_0^\beta \bar m(t)\,\textrm{d}t +J\sqrt{\alpha \int_0^\beta \bar p (t)\,\textrm{d}t}\Big],\\
	\bar p(t') &= \frac{\int_0^\beta \bar q(t) \, \textrm{d}t}{\Big(1-\int_0^\beta [1-\bar q(t)]\,\textrm{d}t\Big)^2}.
	\label{hop_functionalSCE}
	\ees
	\begin{proof}
		As anticipated, by setting the interpolating parameters as $t=\beta$ and $x,y,z,w=0$, we recover the Hopfield picture, thus we get \eqref{hop_solution} by evaluating \eqref{hop_mechanicalsolution} in that point. The derivation of the self-consistency equations is based on the extremization of the functional \eqref{hop_solution} with respect to functional variations of the trajectories $m(t), q(t), p(t)$. Therefore, we must set to zero the functional derivatives of \eqref{hop_solution} w.r.t. all the order parameters appearing as functions of the temporal variable $t$. The self-consistency equation for $\bar m(t)$ can be found in the following way:
		\bes
		\fder{A_{\rm RS}}{\bar m(t')}&=\bbE_J \tanh\Big[\int_0^\beta \bar m(t)\,\textrm{d}t +J\sqrt{\alpha \int_0^\beta \bar p (t)\,\textrm{d}t}\Big]\int_0^\beta \fder{\bar m(t)}{\bar m(t')}\,\textrm{d}t-\int_0^\beta \bar m(t)\fder{\bar m(t)}{\bar m(t')}\,\textrm{d}t=\\
		&=\bbE_J \tanh\Big[\int_0^\beta \bar m(t)\,\textrm{d}t +J\sqrt{\alpha \int_0^\beta \bar p (t)\,\textrm{d}t}\Big]\int_0^\beta \delta(t-t')\,\textrm{d}t-\int_0^\beta \bar m(t)\delta(t-t')\,\textrm{d}t=\\
		&=\bbE_J \tanh\Big[\int_0^\beta \bar m(t)\,\textrm{d}t +J\sqrt{\alpha \int_0^\beta \bar p (t)\,\textrm{d}t}\Big]-\bar m(t')=0,
		\ees
		thus we have
		\beq
		\bar m(t')=\bbE_J \tanh\Big[\int_0^\beta \bar m(t)\,\textrm{d}t +J\sqrt{\alpha \int_0^\beta \bar p (t)\,\textrm{d}t}\Big].
		\label{hop_magnSCEfunctionalform}
		\eeq
		The remaining self-consistency equations can be obtained by following the same lines and are uniquely determined. On the other hand, as well known, the self-consistent equations can display multiple solutions corresponding to neural configurations which are not necessarily stable as they may possibly be minima for the pressure.
	\end{proof}
\end{Theorem}
\begin{Corollary} The replica-symmetry quenched pressure \eqref{hop_solution} returns the standard AGS theory:
	\bes
	A_{\rm RS}(\bar m,\bar p,\bar q)=&\log 2 + \bbE_J \log\cosh\Big[\beta \bar m+J\sqrt{\alpha \beta \bar p}\Big]-\frac{\beta}{2}[\alpha \bar p\big(1-\bar q\big) +\bar m^2 ]+\\
	&+\frac{\alpha}{2}\frac{\beta \bar q}{1-\beta [1-\bar q]} -\frac{\alpha}{2}\log\Big(1-\beta [1-\bar q]\Big).
	\label{hop_agssolution}
	\ees
	As a consequence, the self-consistency equations \eqref{hop_functionalSCE} recovers the standard AGS result:
	\bes
	\bar p&= \frac{\beta \bar q}{\Big(1-\beta(1-\bar q)\Big)^2},\\
	\bar m&=\bbE_J \tanh\Big[\beta \bar m+J\sqrt{\alpha\beta \bar p}\Big],\\
	\bar q&=\bbE_J \tanh^2\Big[\beta \bar m+J\sqrt{\alpha\beta \bar p}\Big].
	\label{hop_agsSCE}
	\ees
	\begin{proof}
		By inspecting \eqref{hop_functionalSCE}, we can see that the r.h.s.'s do not depend on the temporal coordinate $t'$. Thus, \eqref{hop_functionalSCE} actually predict constant functional order parameters, i.e.
		\beq
		\bar m(t) = \bar m, \quad \bar q(t) = \bar q, \quad \bar p(t) = \bar p\quad \forall t \in \bbR^+.
		\eeq
		Thus, by replacing these values in \eqref{hop_solution} and in \eqref{hop_functionalSCE} we recover the standard AGS picture.
	\end{proof}
\end{Corollary}
\section{The relativistic Hopfield network}\label{definizioni}

The ``relativistic'' Hopfield model has been introduced in \cite{Albert1}, where its investigation was restricted to the low storage (see also  \cite{Notarnicola,Marullo}).
The appellation ``relativistic'' can be understood via the mechanical analogy: the Hamiltonian of the classical Hopfield model \eqref{hop_hbare} in terms of the Mattis magnetizations reads as $H_N(\boldsymbol \sigma| \boldsymbol \xi) = - N \sum_{\mu=1}^{P} m_{\mu}^2/2$ and its statistical mechanical picture can be mapped into the dynamics of a fictitious particle of unitary mass living in a $(P+1)$-dimension  Minkowsky space \cite{Albert1}; in this analogy the Mattis magnetization plays the role of the momentum, and the Hopfield Hamiltonian \eqref{hop_hbare} reads as the (classical) kinetic energy for this particle. As the underlying metric of the space-time is Minkowskian, it is quite natural to generalize the classical expression of the kinetic energy into its relativistic counterpart.
\newline
Here, our interest in the ``relativistic'' Hopfield model is two-fold: on the one hand it provides an extension of the classical model which includes also $p$-body interactions and, in this sense, it constitutes an example of dense networks, on the other hand, it allows us to test our generalized mathematical techniques.
\bigskip
\begin{Definition}
	Set $\alpha, \lambda \in \mathbb{R}^+$  and let $\boldsymbol \sigma \in \{- 1, +1\}^N$ be a configuration of $N$ binary neurons. Given $P=\alpha N$ random patterns $\{\boldsymbol \xi^{\mu}\}_{\mu=1,...,P}$, each made of $N$ digital entries identically and independently drawn with probability $P(\xi_i^{\mu}=+1)=P(\xi_i^{\mu}=-1)=1/2$, for $i=1,...,N$, the ``relativistic'' Hopfield cost-function (or ``Hamiltonian'' to preserve a physical jargon) is
	\beq
	H^{\textrm{rel}}_N(\boldsymbol \sigma| \boldsymbol \xi,\lambda) := H_N(\boldsymbol \sigma| \boldsymbol \xi,\lambda) = -\frac{N}{\lambda}\sqrt{1 + \lambda\sum_{\mu=1}^{P} \Big( \sum_{i=1}^N \frac{\xi_i^{\mu} \sigma_i}{N} \Big)^2}.
	\label{HopfieldRelativ}
	\eeq
\end{Definition}
%
	The additional parameter $\lambda$ is such that for $\lambda =1$ we recover the pure relativistic scenario, while for $\lambda \neq 1$ we get a neural network with enriched computational skills with respect to the classical one (\ref{hop_hbare}); this point has been extensively addressed in \cite{Albert1}, while here we focus on the mathematical backbone of the theory.
%
\begin{Remark}
	By expanding the relativistic model (\ref{HopfieldRelativ}) we obtain
	\begin{equation}\label{espansione}
	\frac{H_N(\boldsymbol \sigma| \boldsymbol \xi,\lambda)}{N} \sim - \frac{1}{\lambda} -\frac{1}{2N^2}\sum_{i,j}^{N,N}(\sum_{\mu=1}^{P}\xi_i^{\mu}\xi_j^{\mu})\sigma_i\sigma_j + \frac{\lambda}{8N^4}\sum_{i,j,k,l}^{N,N,N,N}(\sum_{\mu=1}^{P}\xi_i^{\mu}\xi_j^{\mu})(\sum_{\nu=1}^{P}\xi_k^{\nu}\xi_l^{\nu})\sigma_i\sigma_j\sigma_k \sigma_l + ...
	\end{equation}
	namely a series of (denser and denser) associative networks. 
	Note also the alternate signs of the terms making up the series: the attractive ones contribute to memory storage, while the repulsive ones contribute to memory erasure \cite{Albert1,Albert2,unlearning1}.
\end{Remark}
\begin{Remark}
	In order to get more intuition on the model, we can look at the internal field $h_i$ acting on neuron $i$ and which, in a dynamical picture, can be used to ascertain whether the neuronal state $\sigma_i$ is likely to be updated or not (see e.g., \cite{Coolen, Hertz}). By definition, $H_N = -\sum_i h_i \sigma_i$, in such a way that the internal field reads as
	\beq \notag 
	h_i = \left( \sum_{\nu=1}^P \xi_i^{\nu} m_{\nu} \right) \sum_{k=0}^{\infty} \calB\left(\frac{1}{2},k\right) \lambda^{k-1} \left ( \sum_{\mu=1}^P m_{\mu}^2 \right)^{k-1} = - \frac{H_N(\boldsymbol \sigma| \boldsymbol \xi,\lambda)}{N} \left( \sum_{\nu=1}^P \xi_i^{\nu} m_{\nu}\right) \left ( \sum_{\mu=1}^P m_{\mu}^2 \right)^{-1},
	\eeq
	where $m_\mu=\frac1N \sum_i \xi^\mu_i \sigma_i$ is the Mattis magnetization associated to the pattern $\xi^\mu$, while $\calB (n,k)$ is the general binomial coefficient, namely $\calB (n,k) = \Gamma(n+1)/[\Gamma(k+1)\Gamma(1-k+n)]$, $n,k \in \mathbb{R}$. One can see that configurations corresponding to simultaneously large magnitudes of the magnetization entries  (i.e., $\sum_{\mu=1}^P m_{\mu}^2$ large) yield to weaker fields; also, by increasing $\lambda$ the internal field decreases and the system gets more susceptible to the effects of external noise sources.
\end{Remark}
%
%
%
%
\begin{Definition}
	The infinite volume limit of the quenched pressure of the ``relativistic'' Hopfield network, $A^{\textrm{rel}}(\alpha,\beta,\lambda)$, can be written in terms of the partition function of the model $Z^{\textrm{rel}}_{N}(\boldsymbol \xi,\beta,\lambda)$ as
	\bea\label{freeE}
	A^{\textrm{rel}}(\alpha,\beta,\lambda) &:=& A(\alpha,\beta,\lambda) = \lim_{N \to \infty} A_{N}(\alpha, \beta, \lambda) = \lim_{N \to \infty} \frac{1}{N}\mathbb{E}\ln Z_{N}(\boldsymbol \xi,\beta,\lambda),\\
	Z^{\textrm{rel}}_{N}(\boldsymbol \xi,\beta,\lambda)&:=&  Z_{N}(\boldsymbol \xi,\beta,\lambda) := \sum_{\{ \boldsymbol \sigma\}}\exp\left(-\beta H_N(\boldsymbol \sigma | \boldsymbol \xi,\lambda) \right).
	\eea
\end{Definition}
As in the previous section, we use the subscript $N$ to stress that we are working at finite size, while when we omit it we mean that we are evaluating quantities in the thermodynamic limit.
\begin{Definition}
	Given $O(\boldsymbol \sigma)$ as a generic function of the neuron configuration, we define the Boltzmann average $\omega_{\boldsymbol{\xi}}(O)$, its replicated product state over $s$ replicas $\Omega$ and its quenched expectation $\langle O \rangle$ respectively as
	\begin{eqnarray}
	\omega_{\boldsymbol{\xi}}(O) &:=& \frac{ \sum_{\{ \boldsymbol \sigma\} } O(\boldsymbol \sigma)e^{-\beta H_N(\boldsymbol \sigma| \boldsymbol \xi,\lambda)}}{Z_{N}( \boldsymbol \xi,\beta,\lambda)}, \\
	\Omega(O) &:=& \omega_{\boldsymbol{\xi}}^{(1)}(O)\times \omega_{\boldsymbol{\xi}}^{(2)}(O) \times... \times \omega_{\boldsymbol{\xi}}^{(s)}(O),\\
	\label{eq:av}
	\langle O \rangle &:=& \mathbb{E} \Omega(O).
	\end{eqnarray}
\end{Definition}
In the next section, we will exploit a novel interpolation strategy to obtain an explicit expression for the quenched pressure of this model, then, in Sec. \ref{momentosecondo}, preserving the same interpolation (but focusing on the variances of the order parameters instead of the pressure), we will provide a detailed picture of its (possible) critical behavior.

\section{The mechanical generalization of Guerra's interpolation scheme}\label{momentoprimo}

In this Section we achieve the RS expression for the quenched pressure of the relativistic Hopfield model again exploiting a differential-equation-based approach. The main problem consists in the fact that, unlike the classical counterpart, here the cost function (\ref{espansione}) is not monomial in its degrees of freedom $\{\sigma_i\}_{i=1,...,N}$, but it is a (convergent) infinite series. In order to deal with this, we take advantage of the integral representation of the Dirac delta directly on the cost function (\ref{HopfieldRelativ}), rather than on the order parameters as usual \cite{Coolen}. This is represented by the following
\begin{Proposition}
	The partition function of the ``relativistic'' Hopfield model (\ref{HopfieldRelativ}) displays the following integral representation
	\bes \label{eq:integral}
	Z_{N}(\boldsymbol \xi,\beta,\lambda)\propto&  \sum_{\{ \boldsymbol \sigma \}} \int dX\,dK\,\calD\tau\, \exp \Bigg [\frac{\beta N}{\lambda}\sqrt{1 + \lambda X} -\frac{KX \beta N}{2} +    \\
	&\hspace{10em}  +\frac{\beta N}{2} \Big( \frac{\sqrt{K}}{N}\sum_{i=1}^N \xi^1_i \sigma_i \Big)^2+\sqrt{\frac{\beta K}{N}}\sum_{\mu=2}^{P} \sum_{i=1}^N \xi_i^{\mu}\sigma_{i}\tau_{\mu}\Bigg ].
	\ees

\end{Proposition}
\begin{proof}
	The proof works by direct construction. In fact, by definition, the partition function for the Hamiltonian (\ref{HopfieldRelativ}) is
	\bes
	Z_{N}(\boldsymbol \xi,\beta,\lambda)=& \sum_{\{ \boldsymbol \sigma \}} \exp\left[\frac{\beta N}{\lambda}\sqrt{1+\frac{\lambda}{N^2}\sum_{\mu=1}^P \Big( \sum_{i=1}^{N}\xi_i^{\mu}\sigma_i \Big)^2} \right],
	\ees
	and, expanding the square and exploiting the integral representation of the Dirac delta, we get
	\bes
	Z_{N}(\boldsymbol \xi,\beta,\lambda)=& \sum_{\{ \boldsymbol \sigma \}}  \int dX \delta \big( X - \frac{1}{N^2}\sum_{\mu,i,j=1}^{P,N,N}\xi_i^{\mu}\xi_j^{\mu}\sigma_i \sigma_j \big)\exp\Big( \frac{\beta N}{\lambda}\sqrt{1 + \lambda X} \Big)\\
	\propto& \sum_{\{ \boldsymbol \sigma \}} \int dX\,dK\,\exp\Big( \frac{\beta N}{\lambda}\sqrt{1 + \lambda X} + iKX -i \frac{K}{N^2}  \sum_{\mu,i,j=1}^{P,N,N}\xi_i^{\mu}\xi_j^{\mu}\sigma_i \sigma_j\Big).
	\label{ZetaUno}
	\ees
	Notice that the auxiliary variable $X$ allows us to move the sum out of the square root and the auxiliary variable $K$ allows us to recast the $\delta$ function in a mathematically more convenient expression.\\
	Now, we rescale $K \to i\frac{\beta N}{2} K$ as standard \cite{Coolen} and we apply a Hubbard-Stratonovich transformation, in such a way that eq. (\ref{ZetaUno}) can be recast as\footnote{Notice that the complex shift for the variable $K$ formally converts fluctuating functions in exponential functions. However, this does not lead to divergences of the whole partition function, since also integration bounds change: the $K$ variable is therefore integrated along the imaginary axis, so functions in $K$ are still fluctuating.}
	\bes\label{eq:ZetaDue}
	Z_{N}(\boldsymbol \xi,\beta,\lambda) \propto& \sum_{\{ \boldsymbol \sigma \}} \int dX\,dK\,\exp\Big( \frac{\beta N}{\lambda}\sqrt{1 + \lambda X} -\frac{KX \beta N}{2}+ \frac{K \beta}{2N}  \sum_{\mu,i,j=1}^{P,N,N}\xi_i^{\mu}\xi_j^{\mu}\sigma_i \sigma_j\Big)\\
	=&  \sum_{\{ \boldsymbol \sigma \}} \int dX\,dK\,\calD\tau\, \exp \Bigg [\frac{\beta N}{\lambda}\sqrt{1 + \lambda X} -\frac{KX \beta N}{2} +    \\
	&\hspace{10em}  +\frac{\beta N}{2} \Big( \frac{\sqrt{K}}{N}\sum_{i=1}^N \xi^1_i \sigma_i \Big)^2+\sqrt{\frac{\beta K}{N}}\sum_{\mu=2}^{P} \sum_{i=1}^N \xi_i^{\mu}\sigma_{i} \tau_{\mu}\Bigg ],
	\ees
	where, we recall, $\calD\tau : = \prod_{\mu=1}^P \frac{e^{- \tau_{\mu}^2/2}}{ \sqrt{2 \pi}}$.
	Here, under the assumption that a single pattern (say $\boldsymbol \xi^1$) is candidate for retrieval, we split the signal term from the (quenched) noise stemming from all the remaining patterns ($\mu>1$). We also disregarded pre-factors (due to the integral form of Dirac's delta and to rescaling) linear in $N$ and whose contribution in the intensive pressure is vanishing in the thermodynamic limit.
\end{proof}
\begin{Remark}
Here, $X$ and $K$ play the role of auxiliary order parameters. In particular, by construction (notice the delta function), the parameter $X$ is fixed to the value
\begin{equation}
X\equiv \frac{1}{N^2}\sum_{\mu,i,j=1}^{P,N,N}\xi_i^{\mu}\xi_j^{\mu}\sigma_i \sigma_j=\Big(\frac{1}{N}\sum_{i=1}^{N}\xi_i^{1}\sigma_i\Big)^2+\frac{1}{N^2}\sum_{\mu\ge 2}^P\sum_{i,j=1}^{N,N}\xi_i^{\mu}\xi_j^{\mu}\sigma_i \sigma_j,
\end{equation}
where the signal and the noise terms are highlighted. The fact that these display the same form as in the Hopfield model suggests that even in the ``relativistic'' model the number of storable patterns is $P\sim \alpha N$.
\end{Remark}

\begin{Definition} \label{def:rop}
	Besides the auxiliary order parameters $X$ and $K$, the order parameters needed to handle the model (\ref{HopfieldRelativ}) are the natural relativistic extension of the standard ones, namely the Mattis magnetization and the two-replica overlaps introduced in Definition \ref{hop_orderparameters},  defined as follows:
	\bea \label{orderparameters22}
	\tilde{m}_1(\boldsymbol \sigma) &:=& m_1(\boldsymbol \sigma | \boldsymbol \xi) :=\frac{1}{N}\sum_{i=1}^N \xi^1_i \sigma_i,\\
	\label{orderparameters23}
	\tilde{q}_{12} (\boldsymbol \sigma)&:=&\frac{1}{N}\sum_{i=1}^N \sigma^{(1)}_i\sigma^{(2)}_i,\\
	\label{orderparameters_end}
	p_{12}(\boldsymbol \tau)&:=&\frac{1}{P-1}\sum_{\mu>1}^{P} \tau^{(1)}_\mu \tau^{(2)}_\mu, \quad 	p_{11}(\boldsymbol \tau):=\frac{1}{P-1}\sum_{\mu>1}^{P} (\tau^{(1)}_\mu)^2.
	\eea
\end{Definition}
\begin{Remark}
	The reason for the tilde notation is that, in the following, we will introduce a transformation (by a multiplicative factor) on parameters $\tilde{m}$ and $\tilde{q}$ and the transformed variables shall be denoted as $m$ and $q$, respectively.
Also notice that, unlike the previous approach for the classical Hopfield model, here we need also the diagonal term $p_{11}$, hence we will not drop the indices.
	\end{Remark}%
\begin{Definition}
	Assuming that the order parameters display well defined, replica-symmetric expectations in the thermodynamic limit, we define
	\bea
	\bar {\tilde{m}} &:=&\lim_{N\to \infty} \mv{m_1},\\
	\bar {\tilde{q}}_{12} &:=&\lim_{N\to \infty}  \mv{q_{12}}, \\
	\bar p_{12} &:=&\lim_{N\to \infty} \mv{p_{12}},\\
	\bar p_{11} &:=&\lim_{N\to \infty} \mv{p_{11}}.
	\eea
	Analogously, for the auxiliary order paramaters,
	\bea
	\bar X&:=&\lim_{N\to \infty} \mv{X},\\
	\bar K&:=&\lim_{N\to \infty} \mv{K}.
	\eea
\end{Definition}

We can now discuss our strategy to solve the ``relativistic'' Hopfield model. We anticipate that the main idea is to introduce an interpolating pressure $\mathcal{A}_{N}(t)$ that recovers the original model for $t=1$, while for $t=0$ it corresponds to the pressure of a simpler model analytically addressable; then, the expression for $\mathcal{A}_{N}(t)$ is obtained by exploiting the fundamental theorem of calculus
$$\mathcal{A}_N(t=1)=  \mathcal{A}_N(t=0) + \int_{0}^{1} \dot{\mathcal{A}}_N(\tilde t) \textrm d \tilde t.$$
The expression for $\mathcal{A}_{N}(t)$ can be figured out recalling that it can display a number of additional variables (effective only for $t \neq 1$), which can be set a posteriori; with a suitable choice of these variables, the derivative of $\mathcal{A}_{N}(t)$ could be written in terms of the correlation functions of the order parameters of the original model in such a way that, at least under some assumptions (e.g., replica symmetry), the integral of the derivative of the interpolating pressure can be solved. 
\newline
The preliminary steps in this path are therefore the introduction of an interpolating pressure (Definition \ref{def:inter}), the evaluation of its streaming (Lemma \ref{Propa2}) and of the Cauchy condition (Lemma \ref{Propa3}).

\begin{Definition} \label{def:inter}
	Given a scalar interpolating parameter $t \in [0,1]$, $N+P$ auxiliary quenched i.i.d. random variables $J_i \sim \mathcal{N}[0,1], i \in (1,...,N)$ and $\tilde{J}_{\mu}\sim \mathcal{N}[0,1], \mu \in (1,...,P)$, and five real constants $C_m,\ C_{\sigma},\ C_{z}, V_{\sigma}, V_\tau$ to be set a posteriori, we define the Guerra Generalized Action $\mathcal{A}_N(t)$ as
	\bes\label{GuerrA}
	\mathcal{A}_{N}(t) :=& \frac{1}{N}\mathbb{E}\log \sum_{\{ \boldsymbol \sigma \}} \int dX\,dK\,\calD\tau\exp\Big[\frac{\beta N}{\lambda}\sqrt{1 + \lambda X} -\frac{KX \beta N}{2}+t \frac{\beta N}{2}m_1^2 +\\
	&+ (1-t) C_m N m_1 + \sqrt{t}\sqrt{\frac{\beta}{N}}\sum_{\mu,i=1}^{P,N}\xi_i^{\mu}\sqrt{K}\sigma_i \tau_{\mu} + \sqrt{1-t}\big(C_\sigma\sum_{i=1}^{N}J_i \sqrt{K}\sigma_i +\\
	&+ C_\tau \sum_{\mu=1}^P\tilde{J}_{\mu}\tau_{\mu}\big)+(1-t)V_\tau\sum_{\mu=1}^P \frac{\tau_{\mu}^2}{2}+\frac{(1-t)}{2}V_\sigma N K\Big],
	\ees
	where now $\mathbb{E}$ averages over all the quenched random variables involved in the above expression.
\end{Definition}

\begin{Remark}
In the following, instead of working with $t$-dependent expectations for the order parameters and to check {\it a posteriori} that they do not depend on the interpolation parameter $t$, for the generic observable $X$ we will directly assume that $\lim_{N \to \infty} P_t(X)=\delta (X-\bar X)$ for all $t\in [0,1]$. However, in the Appendix, we check that the results of our interpolation procedure are consistent with the replica trick computations, which do not make use of any interpolation scheme. 
\end{Remark}
To simplify calculations it is convenient to introduce a new set of order parameters, as stated in the following
\begin{Definition}
	The Mattis magnetization and the two-replica overlap defined, respectively, in (\ref{orderparameters22}) and (\ref{orderparameters23}), are transformed as 
	\bea \label{orderparameters}
	m(\boldsymbol \sigma) &:=& {\sqrt{K}}\tilde{m}(\boldsymbol \sigma),\\
	 \label{orderparameters_bo}
	q_{12} (\boldsymbol \sigma)&:=&\sqrt{K^{(1)}K^{(2)}}\tilde q_{12}(\boldsymbol \sigma),
	\eea
	where $K^{(a)}$ means that the variable $K$ is evaluated in the $a$-th replica of the system.
	The corresponding RS expectations are denoted with the upper bar, i.e., $\bar{m}$ and $\bar{q}_{12}$.
	Conversely, the variables $X$, $K$, $p_{11}$ and $p_{12}$ are left unchanged. 
\end{Definition}%

As in the previous section dealing with the classical Hopfield model, we exploit the \emph{universality property} of the quenched noise in the thermodynamic limit and retain the signal as Boolean while the remaining the contribution from the remaining $P-1$ patterns is treated as a Gaussian variable \cite{Univ1,Genovese,Agliari-Barattolo}.

\begin{Lemma}\label{Propa2}
	In the infinite volume limit, the $t$-derivative of the Guerra Generalized Action, under the replica symmetric ansatz, reads as
	\bes
	\frac{d\mathcal{A}_{\RS}(t)}{dt} = &-\frac{\beta}{2} \bar{m}^2  
	-\frac{\beta\alpha}{2}(\bar p_{11}\bar K-\bar p_{12}\bar q_{12}).
	\ees
\end{Lemma}
\begin{proof}
	We evaluate directly the $t$-derivative of $\mathcal{A}_{N}(t)$ and get
	\bes
	\frac{d \mathcal{A}_{N}(t) }{ dt } = &\frac{1}{N}  \Bigg[\frac{\beta N}{2} \mvt{m_1^2} + \frac{\beta}{2N} \sum_{\mu,i}^{P,N} \left(\mvt{K \tau_\mu^2} - \mvt{\sqrt{K}\sigma_i \tau_\mu}^2\right) - C_m N \mvt{m_1}- \frac{V_\tau}{2}\sum_\mu^P \mvt{\tau_\mu^2}
	+ \\
	& - \frac{C_\sigma^2}{2} \sum_i \left(\mvt{K}-\mvt{\sqrt{K}\sigma_i}^2 \right) - \frac{V_\sigma N}{2} \mvt{K} - \frac{C_\tau^2}{2}\sum_u^P \left(\mvt{\tau_\mu^2}-\mvt{\tau_\mu}^2\right)  \Bigg],
	\ees
	where $\mvt{\cdot}$ generalizes (\ref{eq:av}) to the interpolating framework \eqref{GuerrA} and $\langle \cdot \rangle_{t=1} = \mv{\cdot}$.
	Using the definitions \eqref{orderparameters22}-\eqref{orderparameters_end} and (\ref{orderparameters})-(\ref{orderparameters_bo}) for the order parameters, the streaming becomes
	\bes
	\frac{d\mathcal{A}_{N}(t)}{dt} = &\Big[\frac{\beta}{2} \mvt{m_1^2} + \frac{\beta\alpha}{2} \mvt{Kp_{11}} - \frac{\beta\alpha}{2} \mvt{q_{12}p_{12}} - C_m  \mvt{m_1}- \frac{V_\tau\alpha}{2} \mvt{p_{11}} + \\
	& - \frac{C_\sigma^2}{2}\mvt{K}+ \frac{C_\sigma^2}{2}\mvt{q_{12}} - \frac{V_\sigma }{2} \mvt{K} - \frac{\alpha C_\tau^2}{2} \mvt{p_{11}} + \frac{\alpha C_\tau^2}{2} \mvt{p_{12}}  \Big].
	\ees
	Recalling that $m_1 = \bar m + \Delta m$ (see \ref{eq:delta}), and similarly for the other order parameters, we can write 
	\bea
	\nonumber
	\langle m_1^2\rangle_t &=& \langle (\bar{m} + \Delta m)^2 \rangle_t =- \bar{m}^2 + 2 \bar{m}\mv{m_1}_t+ \langle (\Delta m)^2 \rangle_t,  \\
	\nonumber
	\langle q_{12} p_{12} \rangle _t&=& \langle (\bar{q}_{12} + \Delta {q}_{12})(\bar{p}_{12} + \Delta p_{12})\rangle_t = - \bar{q}_{12}\bar{p}_{12} + \bar{p}_{12}\langle q_{12}\rangle_t + \bar{q}_{12}\langle p_{12} \rangle_t + \langle \Delta q_{12} \Delta p_{12} \rangle_t,\\
	\nonumber
	\langle K p_{11} \rangle_t &=& \langle (\bar K + \Delta K)(\bar{p}_{11} + \Delta p_{11})\rangle_t = - \bar K\bar{p}_{11} + \bar{p}_{11}\langle K\rangle_t + \bar K\langle p_{11} \rangle _t+ \langle \Delta K \Delta p_{11}\rangle_t.
	\eea
	Now, under the assumption $\langle \Delta X \rangle_t \to 0$ and $\langle \Delta X  \Delta Y \rangle_t \to 0$ for all the observables involved, the $t$-streaming of the Guerra Generalized Action reads as
	\bes
	\frac{d\mathcal{A}_{\RS}(t)}{dt} = &-\frac{\beta}{2} \bar{m}^2  -\frac{\beta\alpha}{2}(\bar p_{11}\bar K-\bar p_{12}\bar q_{12}) + \Big[\beta \bar{m}\mvt{m_1}+  \frac{\beta\alpha}{2} \bar K \mvt{p_{11}}+  \frac{\beta\alpha}{2}\bar p_{11} \mvt{K} \\
	&- \frac{\beta\alpha}{2}\bar  q_{12}\mvt{p_{12}}+ - \frac{\beta\alpha}{2} \bar p_{12}\mvt{q_{12}}- C_m  \mvt{m_1}- \frac{V_\tau\alpha}{2} \mvt{p_{11}} - \frac{C_\sigma^2}{2}\mvt{K}+ \frac{C_\sigma^2}{2}\mvt{q_{12}} \\
	&- \frac{V_\sigma }{2} \mvt{K} +- \frac{\alpha C_\tau^2}{2} \mvt{p_{11}} + \frac{\alpha C_\tau^2}{2} \mvt{p_{12}}  \Big],
	\ees
	and, by choosing
	\begin{equation}\label{eq:choice}
	C_m=\beta \bar{m},\quad C_{\sigma}=\sqrt{\alpha \beta \bar{p}_{12}},\quad C_\tau=\sqrt{\beta \bar{q}_{12}} ,\quad V_\tau=\beta(\bar K-\bar{q}_{12}),\quad V_\sigma=\beta\alpha (\bar{p}_{11}-\bar{p}_{12}),
	\end{equation}
	we can simplify further the streaming by putting to zero the terms in square brackets, therefore obtaining the statement of the lemma.
\end{proof}
\begin{Lemma}\label{Propa3}
	In the infinite volume limit, the Cauchy condition $\mathcal{A}_{\RS}(t=0)$ of the Guerra Generalized Action, under the replica symmetric ansatz reads as
	\bea
	\nonumber
	\mathcal{A}_{\RS}(t=0) &=&\log2 -\frac{\beta \bar K \bar  X}{2} + \frac{\beta\alpha}{2}\bar K(\bar p_{11}-\bar p_{12}) + \frac{\beta}{\lambda}\sqrt{1+\lambda \bar  X} + \bbE_J \log\cosh\Big[\beta \bar m \sqrt{\bar K} + J  \sqrt{\beta\alpha \bar  p_{12} \bar K }\Big] \\
	\label{Cauchy}
	&-&\frac{\alpha}{2}\log\big[1-\beta(\bar K-\bar q_{12})\big]+\frac{\alpha}{2}\frac{\beta \bar q_{12}}{1-\beta(\bar K-\bar q_{12})}.
	\eea
\end{Lemma}
\begin{proof}
	The 1-body term we need to evaluate is
	\bes
	\mathcal{A}_{RS}(t=0) :=&\lim_{N \to \infty} \frac{1}{N}\mathbb{E}\log \sum_{\{ \boldsymbol \sigma \}} \int dX\,dK\,\calD\tau\exp\Big[\frac{\beta N}{\lambda}\sqrt{1 + \lambda X} -\frac{KX \beta N}{2}+\frac{1}{2}V_\sigma N K \\
	&+  C_m N m_1 + \big(C_\sigma\sum_{i=1}^{N}J_i \sqrt{K}\sigma_i + C_\tau \sum_{\mu=1}^P \tilde{J}_{\mu}\tau_{\mu}\big)+V_\tau\sum_{\mu=1}^P \frac{\tau_{\mu}^2}{2}\Big],
	\ees
	where we recall that the values of the constants $C_m,\ C_{\sigma},\ C_{z},\ V_{\sigma},\ V_\tau$ are fixed by the condition \eqref{eq:choice}. By inspecting this quantity, we can see that the $\boldsymbol \sigma$-dependent terms are linear in the spins, while the integral over the variables $\boldsymbol \tau$ is of Gaussian type. Therefore, the $(\boldsymbol \sigma, \boldsymbol \tau)$-dependent part can be easily evaluated, with the result being a non-trivial function of $X$, $K$ and the other order parameters. However, the remaining integral (over the variables $K$ and $X$) is of Laplace form, so the leading behaviour in the thermodynamic limit can be computed by applying the saddle point method. By following this route, with straightforward computations, we reach the thesis (\ref{Cauchy}).
\end{proof}
\begin{Theorem}
	The infinite volume limit of the replica symmetric expression of the quenched pressure of the ``relativistic'' Hopfield model in the high-storage regime reads as
	\bes
	A_{\RS}(\alpha,\beta,\lambda) =&\log2 -\frac{\beta K}{2} {m}^2  -\frac{\beta\alpha K}{2}p(1-q) -\frac{\beta K X}{2}  + \frac{\beta}{\lambda}\sqrt{1+\lambda X} +\frac{\alpha}{2}\frac{\beta K q}{1-\beta K (1-q)} + \\
	&+\bbE_J \log\cosh\Big[\beta m K + J \sqrt{\beta\alpha p K }\Big] -\frac{\alpha}{2}\log\big[1-\beta K (1-q)\big],
	\label{main-res}
	\ees
	whose extremization returns the following self-consistencies for the order parameters
	\bea
	\label{eq:sc1}
	K &=& \frac{1}{\sqrt{1+\lambda X}}, \\
	X &=& {m}^2 + \alpha{p}(1-{q}) + \frac{\alpha}{1-K\beta (1-{q})}, \\
	{q} &=& \mathbb{E}_J \tanh^2\left( K \beta {m} + J \sqrt{K \alpha \beta {p}} \right), \\
	{p} &=& \frac{K\beta{q}}{[1-K \beta (1-{q})]^2}, \\
	\label{eq:sc2}
	{m} &=& \mathbb{E}_J \tanh\left( K \beta {m} + J \sqrt{K \alpha \beta {p}} \right).
	\eea
\end{Theorem}
\begin{proof}
	The explicit expression of the infinite volume limit of the quenched pressure of the ``relativistic'' Hopfield network, in terms of the natural order parameters of the theory, can be obtained via the fundamental theorem of Calculus as
	\beq
	A(\alpha,\beta,\lambda)= \left. \lim_{N \to \infty} \mathcal{A}_{N}(t=1)=\lim_{N \to \infty}\left( \mathcal{A}_{N}(t=0) + \int_{0}^{1}\frac{d\mathcal{A}_{N}}{dt} \right|_{t=t'}dt' \right).
	\label{sumrule}
	\eeq
	By using the sum rule \eqref{sumrule} and taking into account Proposition \ref{Propa2} and Proposition \ref{Propa3}, in the thermodynamic limit we can write
	\bes
	\label{eq:main-res}
	A(\alpha,\beta,\lambda) =&\log2 -\frac{\beta}{2} \bar{m}^2  -\frac{\beta\alpha}{2} \bar p_{12}(K-\bar q_{12}) -\frac{\beta \bar K \bar X}{2}  + \frac{\beta}{\lambda}\sqrt{1+\lambda \bar X}  -\frac{\alpha}{2}\log\big[1-\beta(\bar K-\bar q_{12})\big] \\
	&+\bbE_J \log\cosh\Big[\beta \bar m \sqrt{\bar K} + J \sqrt{\beta\alpha \bar p_{12} \bar K }\Big]+\frac{\alpha}{2}\frac{\beta \bar q_{12}}{1-\beta(\bar K-\bar q_{12})}.
	\ees
Finally, to restore the original order parameters \eqref{orderparameters22}, we perform a rescaling $\bar q_{12}  = \bar{K} \bar{\tilde{q}}$, $\bar m = \sqrt{\bar K} \bar{\tilde{m}}$, therefore obtaining Eq. \ref{main-res} (notice that we are omitting the upper bar and tilde to lighten the notation). The self-consistency equations follow by extremization of the pressure over the order parameters as standard.
\end{proof}
\begin{Remark}
	By forcing $\alpha=0$ in the eq. (\ref{eq:main-res}) we recover the low-storage scenario of the relativistic Hopfield network \cite{Albert1,Notarnicola}, as expected.
\end{Remark}
\begin{Remark}
	As it happens in Mechanics when performing the classical limit of a relativistic theory, here the classical limit of the quenched free energy of the relativistic Hopfield model collapses as well to the classical Hopfield model addressed in Sec. \ref{riscaldamento}.
\end{Remark}


\section{Fluctuation theory, criticality and ergodicity breaking}\label{momentosecondo}

A main contribution of the statistical-mechanics approach to neural networks is the synthesis of their rich phenomenology, as parameters are varied, in terms of a phase diagram. In particular, the Hopfield model is known to exhibit an \emph{ergodic} (E) phase where fast noise prevails and the neuron state is random,  a \emph{spin-glass} (SG) phase where the slow noise due to pattern interference prevails and the neuron state gets stuck in spurious states, and a \emph{retrieval} phase where the system can work as an associative memory and perform pattern recognition; the retrieval phase can be further split into a sub-region (R) where the retrieval state corresponds to a global minimum of the free energy and another sub-region (MR) where the retrieval  state corresponds to a local minimum of the free energy  (see Fig.~\ref{fig:transition}). These regions can be distinguished by looking at the expectation value of the order parameters: in the ergodic region the expectation of all the order parameters is zero, in the spin-glass region the expectation of the Mattis magnetization is zero while the expectation of the overlap is non-null, and in the retrieval region the expectation of the Mattis magnetization and of the overlap is non null.
The transition between the retrieval and the spin-glass phase is first-order, as evidenced by the abrupt drop of the Mattis magnetization, while the transition between the spin-glass phase and the ergodic phase is second-order (or critical), as evidenced by the continuous drop of the overlap and by the divergence of their fluctuations.
In fact, in pairwise models the onset of ergodicity is typically signalled by a critical behavior whose analytical investigation is feasible by a fluctuation analysis. On the other hand, in pure $p$-spin models (e.g., a $p=4$ spin glass) the onset of ergodicity is first-order and the study of this kind of transition is much more challenging.

The relativistic model we are addressing in this paper is neither a pure pairwise model nor a pure $p$-spin model, whence the quest for a deep analysis of ergodicity breaking, which will be addressed hereafter. In order to accomplish this task, we focus on the two-replica overlaps and, once centered them around their expectations, we study the evolution of their variances (suitably amplified by their volumes, see Def. \ref{riscalaggio}) in the space of the tunable parameters $\{ \alpha, \beta, \lambda \}$ to inspect if and where these diverge, eventually  marking the onset of criticality.
To this goal we retain the interpolation defined by the Guerra Generalized Action (\ref{GuerrA}) and we use it to evaluate the expectations of the order parameters fluctuations and correlations at $t=1$. As before, to achieve this result, we start the evaluation at $t=0$ (where a fictitious, but mathematically treatable, environment is experienced by the neurons) and then propagate this solution up to $t=1$ (where the real surrounding is perceived by the neurons).
\begin{Definition}\label{riscalaggio}
Using $(l,m)$ to label replicas, the centered and rescaled overlap fluctuations $\theta_{lm}$ and $\rho_{lm}$ are introduced as
\bes
\theta_{lm}&:=\sqrt{N}\big[q_{lm}-\delta_{lm}Q-(1-\delta_{lm})q\big],\\
\rho_{lm}&:=\sqrt{P}\big[p_{lm}-\delta_{lm}P'-(1-\delta_{lm})p\big],
\ees
where, $Q$ and $q$ are the expectation values of, respectively, $q_{11}$ and $q_{12}$, and, analogously, $P'$ and $p$ are the expectation values of, respectively, $p_{11}$ and $p_{12}$.
\end{Definition}
\begin{Proposition} Given $O$ as a smooth function of overlaps $\{ q_{lm}, p_{lm}\}_{l,m=1,...,s}$ involving $s$ distinct replicas, the following streaming equation holds
\beq
d_{T}  \mv{O}=\frac{1}{2}\sum _ { a , b } ^ { s} \mv {O \cdot g_ { a , b }}- s \sum _ { a = 1 } ^ { s } \mv{O \cdot g _ { a , s + 1 } } +  \frac{s ( s+ 1 )}{2} \mv{ O \cdot g_ { s+ 1 , s + 2 } } -\frac{s}{2} \mv{ O \cdot g_ { s + 1 , s + 1 } },\eeq
\label{eq:fluctstream}
where we posed
$$
g_{a,b} = \theta_{ab} \rho_{ab}
$$
and
$$d_T := \frac{1}{\beta\sqrt{\alpha}} \frac{d}{dt}$$
to simplify notation.
\end{Proposition}
We will not report the proof of this proposition as passages are quite lengthy but rather standard (see \cite{Barra-JSP2010,Alemannation1}), rather, to understand the physics underlying this streaming, it is enough to show how the core of these rules can be obtained in general. The simplest (already $s$-replicated) interpolating structure reads as \beq
Z^s(t)=\sum_{\{ \sigma^{(a)} \} }\int d \pi (K^{(a)})\, e^{-\beta \sum_{a=1}^{s}H(\sqrt{K^{(a)}},\sigma^{(a)})}e^{\sqrt{\frac{t}{N}}\sum_{a,i=1}^{s,N}J_i \sqrt{K^{(a)}} \sigma_i^{(a)}}
\eeq
where $d \pi (K^{(a)})$ denotes a generic continuous distribution for $K^{(a)}$, for instance, in (\ref{eq:integral}) the measure is the complex exponential.
Then, the streaming of the generalized expectaction for a generic smooth function $O$ of the overlaps among $s$ replicas is
\bes
  \partial_t \langle O \rangle &= \partial_t \mathbb{E}\frac{\sum_{\sigma^{(a)}}\int d \pi (K^{(a)}) O e^{-\beta \sum_{a=1}^{s}H(\sqrt{K^{(a)}},\sigma^{(a)})}e^{\sqrt{\frac{t}{N}}\sum_{a,i=1}^{s,N}J_i \sqrt{K^{(a)}} \sigma_i^{(a)}}}{\sum_{\sigma^{(a)}} \int d\pi(K^{(a)}) e^{-\beta \sum_{a}^{s}H(\sqrt{K^{(a)}},\sigma^{(a)})}e^{\sqrt{\frac{t}{N}}\sum_{a,i}^{s,N}J_i \sqrt{K^{(a)}} \sigma_i^{(a)}}} \\
  &= \frac{1}{2\sqrt{tN}}\mathbb{E}\sum_{i=1}^N J_i \sum_{a=1}^{s}\left[\Omega_t(O \sqrt{K^{(a)}}\sigma_i^{(a)})- \Omega_t(O)\Omega_t(\sqrt{K^{(a)}}\sigma_i^{(b)})\right]\\
  &= \frac{1}{2\sqrt{tN}}\mathbb{E}\sum_{i=1}^N \partial_{J_i} \sum_{a=1}^{s}\left[\Omega_t(O \sqrt{K^{(a)}}\sigma_i^{(a)})- \Omega_t(O)\Omega_t(\sqrt{K^{(a)}}\sigma_i^{(a)})\right]\\
    &= \frac{1}{2N}\mathbb{E}\sum_{i=1}^{N}\sum_{a,b=1}^{s,s}\Big[\Omega_t(O\sqrt{K^{(a)}K^{(b)}} \sigma_i^{(a)}\sigma_i^{(b)})-2\Omega_t(O \sqrt{K^{(a)}}\sigma_i^{(a)})\Omega_t(\sqrt{K^{(b)}}\sigma_i^{(b)}) + \\
    &\quad \quad\quad \quad+ 2 \Omega_t(O)\Omega_t(\sqrt{K^{(a)}}\sigma_i^{(a)})\Omega_t(\sqrt{K^{(b)}}\sigma_i^{(b)}) - \Omega_t(O)\Omega_t( \sqrt{K^{(a)}K^{(b)}}\sigma_i^{(a)} \sigma_i^{(b)}) \Big].
\ees
The factorized Boltzmann averages can now be represented as averages over distinct replicas, e.g. $\Omega_t(O)\Omega_t(\sqrt{K^{(a)}K^{(b)}}\sigma_i^{(a)}\sigma_i^{(b)})=\Omega_t(\sqrt{K^{(s+1)}K^{(s+2)}}O \sigma_i^{(s+1)}\sigma_i^{(s+2)})$ for $a \neq b$ and the average $N^{-1}\sum_{i}$ then produces overlaps like $g_{s+1,s+2}$.
\newline

%

We are interested in finding, in the $(\alpha, \beta, \lambda)$ space, the critical surface for ergodicity breaking {\em from the high noise limit} (where no correlations persist) we can treat $\theta_{ab},\rho_{ab}$ as Gaussian variables with zero mean (this allows us to apply Wick theorem inside the averages) and we also treat both $\sqrt{K}\sigma_i$ and $z_\mu$ as zero mean random variables (thus all averages involving uncoupled fields are vanishing): this considerably simplifies the evaluation of the critical surface in the $(\alpha, \beta, \lambda)$ space.
\newline
We can thus reduce the analysis of the rescaled overlap fluctuations to
\bes
\mv{ \theta_{12}^2 }_t &= C(t), &\mv{ \theta_{12}\rho_{12} }_t &= D(t), &\mv{ \rho_{12}^2 }_t &= G(t),\\
\mv{ \theta_{11}^2 }_t &= H(t), &\mv{\theta_{11}\rho_{11} }_t &= I(t), &\mv{\rho_{11}^2}_t &= L(t),\\
\mv{\theta_{11}\rho_{22} }_t &= Q(t),&\mv{\theta_{11}\theta_{22} }_t &= R(t), &\mv{\rho_{11}\rho_{22}}_t &= S(t).
\ees
As stated, the strategy is to evaluate $\mv{ \theta_{12}^2 }_{t=1}$ as well as $\mv{ \theta_{12}\rho_{12} }_{t=1}$ and $\mv{ \rho_{12}^2 }_{t=1}$ by calculating their values at $t=0$ (the Cauchy condition) and then propagating the solution up to $t=1$ again via the Fundamental Theorem of Calculus: the surface, in the $(\alpha,\beta,\lambda)$ space, where these quantities diverge marks the onset of criticality and it is depicted by the next
\begin{Theorem}
The ergodic region of the relativistic Hopfield model in the high-storage is delimited by the following critical surface in the $(\alpha,\beta,\lambda)$ space of the tunable parameters
\beq\label{superficiecritica}
\beta_c=\frac{\sqrt{1 + \sqrt{\alpha} \lambda + \alpha \lambda }}{1 +\sqrt{\alpha}}.
\eeq
\end{Theorem}
\begin{proof}
According to Proposition (\ref{eq:fluctstream}), by inspecting the $t-$evolution of $C(t)=\mv{ \theta_{12}^2 }_t$, $D(t)=\mv{ \theta_{12}\rho_{12} }_t$ and $G(t)=\mv{ \rho_{12}^2 }_t$, we can write down the following system of coupled ODE
\bes
d_T {C}&=2CD, \\
d_T {D}&=D^2+CG,\\
d_T {G}&=2GD.
\label{eq:pde}
\ees
Suitably combining $C$ and $G$ in \eqref{eq:pde} we can write
\beq
d_T \ln \frac{C}{G}=0\implies C(T)=r^2G(T), \quad r^2=\frac{C(0)}{G(0)}.
\eeq
Now we are left with
\bes
d_T{D}&=D^2+r^2G^2,\\
d_T{G}&=2GD.
\label{eq:pde2}
\ees
The trick here is to complete the square by summing $d_T D + r d_T G$, thus obtaining
\bes
d_T Y&=Y^2,\\
Y&=D+rG,\\
d_T{G}&=2G(Y-rG).
\label{eq:pde3}
\ees
The solution of the above system of coupled ODE is trivial and it is given by
\beq
Y(T)=\frac{Y_0}{1- T Y_0},\quad Y_0=D(0)+\sqrt{C(0)G(0)}.
\label{eq:fluctY}
\eeq
Hence we are left with the evaluation of the correlations at $t=0$: namely the Cauchy condition related to the solution coded by eq. (\ref{eq:fluctY}).
To this task we introduce, inside the $\mathcal{A}(t=0)$ term, a one-body generating function of the momenta of $\tau,\sqrt{K}\sigma$ as follows
\bes\label{eq:onebodygenf}
\mathcal{A}(t=0) =\log \sum_{\{ \boldsymbol \sigma \}} \int \calD z \exp\Big[\beta K \sum_{\mu}\tau_{\mu}^2/2 +\sum_i j_i \sqrt{K}\sigma_i + \sum_\mu J_\mu \tau_\mu\Big].
\ees
More precisely, the source fields $j_i$ and $J_\mu$ provide a mathematical tool to generate the average values $\langle (\sqrt{K}\sigma)^P \rangle$ and $\langle z_mu ^P \rangle$ by suitable derivatives of $\calA(t=0)$; for instance, we can generate $\langle \tau_\mu \rangle$ as $\partial_{J_\mu}  \calA(t=0)$ and finally setting $J,j =0$. We stress that this trick is particularly useful because it allows calculating the various mean values necessary to get eq (5.12) by means of simple derivatives of $\calA(t=0)$.
For the sake of clearness, we highlight that the relevant terms in $j,J$ are
\bes
\mathcal{A}(t=0)  \propto \sum_i \log\cosh(j_i \sqrt{K}) + \frac{1}{2(1-\beta K)} \sum_\mu J_\mu^2.
\ees
All the averages needed at $t=0$ can now be calculated simply as its derivatives and {\em after} performing the derivatives by setting $(j=0, J =0)$. This operation leads to
\bes
D(t=0)&=\sqrt{NP}\evalat{\big(\partial_j \mathcal{A}\big)^2\big(\partial_J \mathcal{A}\big)^2}{j,J=0}=0,\\
C(t=0)&=\evalat{\big(\partial^2_j \mathcal{A}\big)^2}{j,J=0}=K^2,\\
G(t=0)&=\evalat{\big(\partial^2_J \mathcal{A}\big)^2}{j,J=0}=(1-\beta K)^{-2}.
\ees
Inserting this result in \eqref{eq:fluctY}, we get
\bes
Y(T)=\frac{\frac{K}{1-\beta K}}{1-T \frac{K}{1-\beta K}}.
\ees
Upon evaluating $Y(T)$ for $T=\beta\sqrt{\alpha} t$ at $t=1$ we obtain
\bes\label{polo}
Y(t=1)=\frac{\frac{K}{1-\beta K}}{1-\beta \sqrt{\alpha} \frac{K}{1-\beta K}}=\frac{K}{1-\beta K(1+ \sqrt{\alpha} )}.
\ees
Since we are interested in obtaining the critical surface where ergodicity breaks down, namely where fluctuations (i.e. $Y(t=1)$) grow arbitrarily large, we can check where the denominator at the r.h.s. of eq. (\ref{polo}) vanishes. This leads to
\bea\label{eq:ergodicityline}
\beta_c &=&\frac{1}{K(1+\sqrt{\alpha})},\\
K &=& \frac{1}{\sqrt{1+\lambda X}},\\
X &=& \frac{\alpha}{1-K\beta_c}.
\eea
The system above can be rearranged explicitly in order to get the critical surface $\beta_c(\alpha,\lambda)$, as stated in the theorem.
\end{proof}
\begin{Remark}
The above expression, for $\lambda=1$ (i.e. in the true relativistic framework), generalizes the standard AGS critical line \cite{Amit} and collapses to the latter in the classical limit as it should.
\end{Remark}
\begin{Remark}
In the standard Hopfield model the critical noise diverges as $T_c:= 1/ \beta_c \sim \sqrt{\alpha}$, as the load $\alpha$ is made larger and larger. 
Conversely, in the relativistic setup, the critical temperature settles on the constant value $T_c:= 1/ \lambda$ for large values of the load $\alpha$. Also notice that, for $\lambda <2$, the critical temperature for ergodicity breaking is a non-monotonous function of $\alpha$: $T_c$ increases as long as $\alpha < (\lambda - 2)^2/2$, then, at $\alpha = (\lambda - 2)^2/2$ it exhibits a maximum, after which the critical temperature slowly decreases towards its asymptotic limit. For $\lambda =2$, the maximum is placed at $\alpha =0$. 
\end{Remark}

Using the fluctuation theory, we were able to highlight a critical behavior and to provide an explicit equation for the critical surface (see eq. \ref{superficiecritica}). Now, we try to get the same result by Taylor-expanding the self-consistency \eqref{main-res} near $q=0$ and $m=0$, in fact, this is a standard procedure in models displaying criticality and here it is used as a consistency check.
By keeping only leading order terms for every order parameter we write
\bea
K &=& \frac{1}{\sqrt{1+\lambda X}}, \\
X &=& \alpha\bar{p}(1-\bar{q}) + \frac{\alpha}{1-K\beta (1-\bar{q})} = \frac{\alpha }{1-\beta  K} + \mathcal O(\bar{q}^2), \\
\bar{q} &=& \mathbb{E}_y \tanh^2\left( y \sqrt{K \alpha \beta \bar{p}} \right) = K\alpha \beta \bar p + \mathcal O(\bar{p}^2) ,\\
\bar{p} &=& \frac{K\beta\bar{q}}{[1-K \beta (1-\bar{q})]^2} = \frac{\beta  K \bar q}{(1-\beta  K)^2} + \mathcal O (\bar{q}^2).
\eea
After some rearrangements we get
\bea
K^2&=& \frac{1-\beta  K}{1-\beta  K+ \alpha\lambda }\\
 \beta  K  &=&\frac{1}{1+ \sqrt{\alpha} }
\eea
whose solution recovers the previous expression for the critical surface, see eq. \eqref{superficiecritica}.
\begin{figure}[tb]
\begin{center}
        \includegraphics[width=0.95\textwidth]{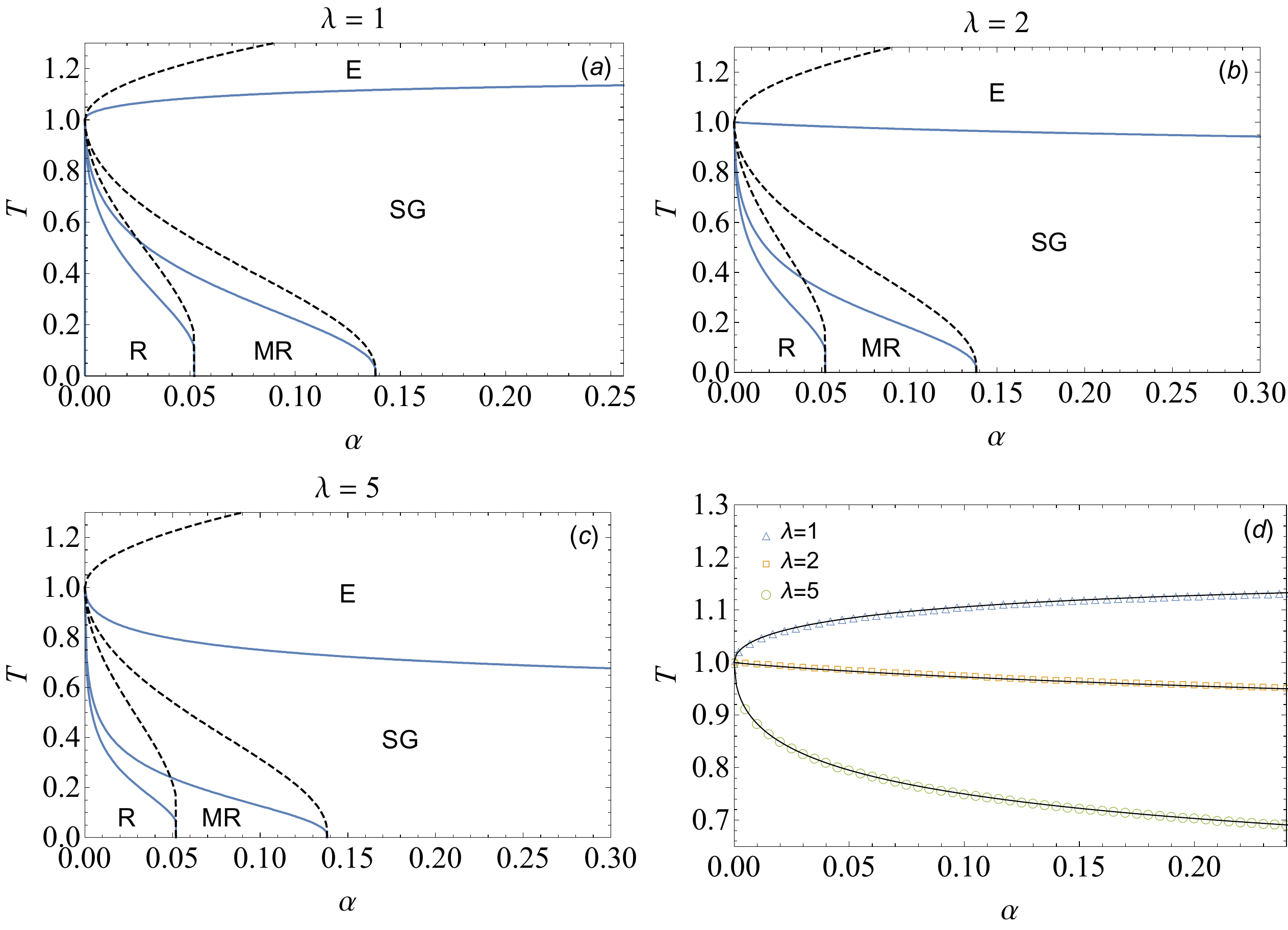}
	\caption{
Phase diagram for the relativistic model for different values of $\lambda$, namely $\lambda=1$ (panel a), $\lambda=2$ (panel b) and $\lambda=5$ (panel c). The phase diagrams are obtained by numerically solving the self-consistency Eqs.~(\ref{eq:sc1})-(\ref{eq:sc1}): the retrieval region (R+MR) is characterized by both non-vanishing $m$ and $q$. The spin-glass phase (SG) is instead given by $m=0$ but non vanishing Edward-Anderson parameter $q$. The distinction between the pure retrieval (R) and mixed retrieval (MR) stands for the fact that, respectively, retrieval and spin-glass states correspond to maxima in the pressure. Finally, in the ergodic region (E) $m=q=0$. Notice that the transition line corresponding to the boundary between the retrieval region (R+MR) and the spin-glass phase (SG) as well as the transition line corresponding to the boundary between the pure retrieval region (R) and the mixed retrieval phase (MR) are close to those found for the classical Hopfield model (black dashed lines) and, as $\lambda$ increases, the retrieval regions shrinks. On the other hand, the critical line signalling the emergence of ergodicity is qualitatively different: the ergodic region engulfs part of the SG region. Notice the, in the $\lambda=1$ phase diagram, we only see the increasing branch of the ergodicity-breaking critical temperature (we decided not to include this behavior in the plot since the critical temperature start to decrease only at $\alpha=1$ for $\lambda=1$, and the retrieval regions would not be transparent by inspecting also the  region). Finally, panel d shows a comparison between the transition lines to ergodicity found with numerical solutions of self-consistency Eqs.~(\ref{eq:sc1})-(\ref{eq:sc1}) (symbols) and the analytical prediction Eq.~\ref{eq:ergodicityline} for different values of $\lambda$ as explained by the legend: the agreement is excellent.}
\label{fig:transition}
\end{center}
\end{figure}
In Fig. \ref{fig:transition}, we reported the phase diagrams of the ``relativistic'' Hopfield model for various values of $\lambda$(=$1$, $2$ and $5$). As a first comment, we see that increasing the parameter $\lambda$, the critical line for the transition to the ergodic region changes convexity, and the ergodic region gets wider towards smaller values of the noise (this features is in common with \cite{Alemannation1}). While the critical storage capacity remains $\alpha_c (\lambda,\beta \to \infty)\sim 0.14$, we see that the corresponding critical line (i.e. the transition to the spin-glass phase) is stretched toward the $\alpha$ axis, meaning that the relativistic model is less stable w.r.t. thermal noise. This can be understood with the fact that the relativistic cost-function presents shallower potential wells w.r.t to the standard Hopfield model, and this feature is amplified for higher and higher values of $\lambda$. This means that less thermal noise is needed for the system to escape from these wells, thus reducing the associative power of the relativistic network.

\section{Conclusions}\label{conclusions}

In this paper we developed mathematical techniques to deal with dense neural networks, possibly without extensively relying upon the statistical mechanics of spin glasses, but rather by exploiting differential-equation theory.
\newline
We first addressed the classical Hopfield network, whose solution was recovered by relying on a {\em mechanical analogy}. The knowledge of statistical mechanics of complex system is rather marginal in this setting and this may trigger further research by applied mathematicians that are more familiar with differential-equation theory but possibly less acquainted with spin glasses.
\newline
In the rest of the paper, we considered the ``relativistic'' Hopfield network as an example of dense network allowing for $p$-body interactions; still by preserving a {\em mechanical setting} we got an explicit expression of its quenched free energy in the high-load regime. We stress that the cost-function of this model is no longer a simple monomial expression of the neurons and the patterns, and handling this cost-function required a non-trivial generalization of the above differential-equation approach.
\newline
At the end of this journey we can finally consider very broad classes of cost functions
\beq
H = \calF\Big[\sum_{\mu=1}^{P} m_{\mu}^2\Big],
\label{HopfieldBare}
\eeq
where $\calF$ is a generic smooth function, and obtain their related pressure (at least under the replica symmetry assumption) and thus all the properties of the model under consideration (e.g., the phase diagram) by following the scheme pursued here.
In fact, for the generic cost function (\ref{HopfieldBare}), the quenched pressure in the infinite volume limit reads as
\bes
A(\alpha,\beta,\lambda) =&\log2 -\frac{\beta K}{2} \bar{m}^2  -\frac{\beta\alpha K}{2}p(1-q) -\frac{\beta K X}{2}  + \frac{\beta}{2}\calF[X] +\frac{\alpha}{2}\frac{\beta K q}{1-\beta K (1-q)} + \\
&+\bbE_y \log\cosh\Big[\beta m K +y \sqrt{\beta\alpha p K }\Big] -\frac{\alpha}{2}\log\big[1-\beta K (1-q)\big],
   \label{main-res2}
\ees
whose extremization returns the following self-consistencies for the order parameters
\bea
K &=&\,\partial_X \calF[X], \\
X &=& \bar{m}^2 + \alpha\bar{p}(1-\bar{q}) + \frac{\alpha}{1-K\beta (1-\bar{q})}, \\
\bar{q} &=& \mathbb{E}_y \tanh^2\left( K \beta \bar{m} + y \sqrt{K \alpha \beta \bar{p}} \right), \\
\bar{p} &=& \frac{K\beta\bar{q}}{(1-K \beta (1-\bar{q}))^2}, \\
\bar{m} &=& \mathbb{E}_y \tanh\left( K \beta \bar{m} + y \sqrt{K \alpha \beta \bar{p}} \right).
\eea
Furthermore, if the leading contribution in (\ref{HopfieldBare}) is provided by pairwise interactions, an ergodic surface is expected and its expression is given by
\beq
\beta =\frac{1}{\left(\sqrt{\alpha }+1\right) \partial_X \calF\left(\alpha +\sqrt{\alpha }\right)}.
\eeq
If $\calF(x)= - Nx/2$ we recover the classical Hopfield scenario, while, if $\calF[x]=-N\sqrt{1+x}$ we recover the relativistic Hopfield scenario, but the method itself can  be applied in full generality to various cost function and can be particularly useful to address dense networks.

\section*{Acknowledgments}

The Authors are grateful to UniSalento, INFN, CNR-Nanotec and Sapienza University for financial support.

\appendix

\section{Relativistic Hopfield solution via replica trick route}\label{app:trick}
In this Appendix, we re-derive the quenched pressure in the thermodynamic limit for the ``relativistic'' Hopfield model by using standard replica trick computations. To do this, we start with representation of the partition function \eqref{eq:ZetaDue}, which we recall here:
\begin{equation}
Z_N (\boldsymbol \xi, \beta, \lambda)= \sum_{ \{ \boldsymbol \sigma \}}\int dXdK\exp\Big(\frac{\beta N}{\lambda}\sqrt{1+\lambda X}-\frac{KX\beta N}2+\frac{K\beta}{2N}\sum_{\mu=1}^P \sum_{i,j=1}^N \xi^\mu_i \xi^\mu _j \sigma_i \sigma_j\Big).
\end{equation}
In the hypotesis of retrieval of a single pattern (say $\boldsymbol {\xi^1}$), we can separate the signal contribution from the noise generated by all the others patterns. Then, the partition function can be rewritten in the form
\begin{equation}
\begin{split}
Z_N (\boldsymbol \xi, \beta, \lambda)= \sum_{\{ \boldsymbol \sigma \} }\int dXdK\exp\Big(\frac{\beta N}{\lambda}\sqrt{1+\lambda X}-\frac{KX\beta N}2&+\frac{K\beta}{2N} \sum_{i,j=1}^N \xi^1_i \xi^1 _j \sigma_i \sigma_j+\\&+\frac{K\beta}{2N}\sum_{\mu>1}^P \sum_{i,j=1}^N \xi^\mu_i \xi^\mu _j \sigma_i \sigma_j\Big).
\end{split}
\end{equation}
We can now linearize the quadratic terms in the $\boldsymbol \sigma$ variables by introducing a set of auxiliary (Gaussian) variables $z_\mu$, so that
\begin{equation}
\begin{split}
Z_N (\boldsymbol \xi, \beta, \lambda)= \sum_{\{ \boldsymbol \sigma \} }&\int dXdK\calD\tau_1\exp\Big(\frac{\beta N}{\lambda}\sqrt{1+\lambda X}-\frac{KX\beta N}2+\sqrt{\frac{K\beta}{N}} \sum_{i=1}^N \xi^1_i \tau_1 \sigma_i \Big)\times\\\times& \int (\prod_{\mu>1} \calD\tau_\mu)\exp\Big(\sqrt{\frac{K\beta}{2N}}\sum_{\mu>1}^P \sum_{i=1}^N \xi^\mu_i \tau_\mu \sigma_i \Big).
\end{split}
\end{equation}
The replica trick approach implies the computation of the quantity $\mathbb E' Z_N^n$, where $n$ is the number of replicas and $\mathbb E'$ is the average over non-recalled patterns, that is
\begin{equation}
\notag
\begin{split}
\mathbb E'Z_N^n&=\mathbb E' \sum_{\boldsymbol \sigma^{(1)}}\dots  \sum_{\boldsymbol \sigma^{(n)}}\int(\prod_{a=1}^n dX^{(a)}dK^{(a)}\calD\tau_1^{(a)})\exp\Big(\frac{\beta N}{\lambda}\sum_a\sqrt{1+\lambda X^{(a)}}-\frac{\beta N}2\sum_a K^{(a)}X^{(a)}\\&+\sqrt{\frac{\beta}{N}}\sum_a \sqrt{K^{(a)}} \sum_{i=1}^N \xi^1_i \tau_1^{(a)} \sigma_i^{(a)} \Big)\times \int (\prod_{\mu>1}\prod_{a=1}^n \calD\tau_\mu^{(a)})\exp\Big(\sqrt{\frac{\beta}{N}}\sum_a \sqrt{K^{(a)}}\sum_{\mu>1}^P \sum_{i=1}^N \xi^\mu_i \tau_\mu^{(a)} \sigma_i^{(a)} \Big),
\end{split}
\end{equation}
where $\calD\tau$ is the usual Gaussian measure. The only term depending on the non-retrieved patterns is the last integral, so we can directly compute the average $\mathbb E'$. Then,
\begin{equation}
\begin{split}
&\int (\prod_{\mu>1}\prod_{a=1}^n \calD\tau_\mu^{(a)})\mathbb E'\exp\Big(\sqrt{\frac{\beta}{N}}\sum_a \sqrt{K^{(a)}}\sum_{\mu>1}^P \sum_{i=1}^N \xi^\mu_i \tau_\mu^{(a)}  \sigma_i^{(a)} \Big)\\&
=\int (\prod_{\mu>1}\prod_{a=1}^n \calD\tau_\mu^{(a)} )\exp\Big(\sum_{\mu>1}\sum_{i=1}^N\log\cosh(\sqrt{\frac{\beta}{N}}\sum_a \sqrt{K^{(a)}}\xi^\mu_i \tau_\mu^{(a)}  \sigma_i^{(a)}) \Big)\\&
\simeq \int (\prod_{\mu>1}\prod_{a=1}^n \calD\tau_\mu^{(a)} )\exp\Big({\frac{\beta}{2N}}\sum_{\mu>1}\sum_{i=1}^N \sum_{a,b=1}^n \sqrt{K^{(a)} K^{(b)}} \tau_\mu^{(a)} \tau_\mu^b \sigma_i^{(a)}\sigma_i^{(b)} \Big),
\end{split}
\end{equation}
where in the last line we used the fact that $\log\cosh x=\frac {x^2}2 +\mathcal O(x^3)$, which is a reasonable approximation since we want to evaluate the thermodynamic limit of the partition function. At this point, we can introduce the overlap order parameters by inserting a Dirac delta in the integral:

\begin{equation}
\notag
\begin{split}
\int (\prod_{\mu>1}\prod_{a=1}^n \calD\tau_\mu^{(a)} )&\exp\Big({\frac{\beta}{2N}}\sum_{\mu>1}\sum_{i=1}^N \sum_{a,b=1}^n \sqrt{K^{(a)} K^{(b)}} \tau_\mu^{(a)} \tau_\mu^{(b)}  \sigma_i^{(a)}\sigma_i^{(b)} \Big)\\ = \int (\prod_{\mu>1}\prod_{a=1}^n \calD\tau_\mu^{(a)} )&\big(\prod_{a,b=1}^n dq_{ab}\delta(q_{ab}-\tfrac1N\sqrt{K^{(a)}K^{(b)}}\sum_{i}\sigma^{(a)}_i \sigma^{(b)}_i)\big)\exp\Big({\frac{\beta}{2N}}\sum_{\mu>1}\sum_{i=1}^N \sum_{a,b=1}^n  \tau_\mu^{(a)} \tau_\mu^{(b)}  q_{ab} \Big)\\
=\int (\prod_{\mu>1}\prod_{a=1}^n \calD\tau_\mu^{(a)} )&\big(\prod_{a,b=1}^n dq_{ab}\frac{Ndp_{ab}}{2\pi}\big)\exp\Big[iN \sum_{a, b}p_{ab}q_{ab}\\&-i\sum_{a,b,i}p_{ab}\sqrt{K^{(a)} K^{(b)}}\sigma^{(a)}_i \sigma^{(b)}_i +{\frac{\beta}{2N}}\sum_{\mu>1}\sum_{i=1}^N \sum_{a,b=1}^n  \tau_\mu^{(a)} \tau_\mu^{(b)}  q_{ab} \Big],
\end{split}
\end{equation}
where in the last line we used the Fourier representation of the Dirac delta by introducing the $q$-conjugated order parameters $p_{ab}$. At this point, we can directly compute the Gaussian integral over the $z$ variables, whose result is
\begin{equation}
\int (\prod_{\mu>1}\prod_{a=1}^n \calD\tau_\mu^{(a)} )\exp\Big({\frac{\beta}{2N}}\sum_{\mu>1}\sum_{i=1}^N \sum_{a,b=1}^n  \tau_\mu^{(a)} \tau_\mu^{(b)}  q_{ab} \Big)=\exp\Big(-\frac{\alpha N}{2}\log \text{det} (\mathbf 1 -\beta \mathbf q)\Big),
\end{equation}
where $\mathbf 1$ is the $n\times n$ identity matrix and $\mathbf q$ is the overlap matrix and we replaced $P-1$ with $\alpha N$, which is possible since we are analyzing the thermodynamic limit. Then, after rescaling $p_{ab}\to \frac{\alpha\beta}{2}p_{ab}$ and $\tau_1^{(a)}\to \sqrt{\beta N} \tau_1^{(a)}$, we have
\begin{equation}
\notag
\begin{split}
\mathbb E'Z_N^n&=\sum_{\boldsymbol \sigma^{(1)}}\dots  \sum_{\boldsymbol \sigma^{(n)}}\int(\prod_{a=1}^n dX^{(a)}dK^{(a)} \sqrt{\frac{\beta N}{2\pi}}d\tau_1^{(a)})\big(\prod_{a,b=1}^n dq_{ab}\frac{i\alpha\beta Ndp_{ab}}{4\pi}\big)\exp\Big(-\frac{\beta N}2 \sum_a (\tau_1^{(a)})^2\\&+\frac{\beta N}{\lambda}\sum_a\sqrt{1+\lambda X^{(a)}}-\frac{\beta N}2\sum_a K^{(a)}X^{(a)}+\beta\sum_a \sqrt{K^{(a)}} \sum_{i=1}^N \xi^1_i \tau_1^{(a)} \sigma_i^{(a)}  -\frac{\alpha \beta N}{2}  \sum_{a,b}p_{ab}q_{ab}\\&+\frac{\alpha\beta}{2}\sum_{a,b,i}p_{ab}\sqrt{K^{(a)} K^{(b)}}\sigma^{(a)}_i \sigma^{(b)}_i-\frac{\alpha N}{2}\log \text{det} (\mathbf 1 -\beta \mathbf q)\Big).
\end{split}
\end{equation}
From now on, we will denote the integration measure over the order parameters simply with $d\mu$ (we stress that the contribution of the prefactors in the integration measures is negligible in the thermodynamic limit). With straightforward computation, it is easy to prove that
\begin{equation}
\begin{split}
\sum_{\boldsymbol \sigma^{(1)}}\dots  \sum_{\boldsymbol \sigma^{(n)}}&\exp\Big(\beta\sum_a \sqrt{K^{(a)}} \sum_{i=1}^N \xi^1_i \tau_1^{(a)} \sigma_i^{(a)}+\frac{\alpha\beta}{2}\sum_{a,b,i}p_{ab}\sqrt{K^{(a)} K^{(b)}}\sigma^{(a)}_i \sigma^{(b)}_i\Big)=\\&
\exp\Big(N \mathbb E \log \sum_{\boldsymbol \sigma}\exp\big(\beta\sum_a \sqrt{K^{(a)}} \xi^1 \tau_1^{(a)} \sigma^{(a)}+\frac{\alpha\beta}{2}\sum_{a,b}p_{ab}\sqrt{K^{(a)} K^{(b)}}\sigma^{(a)} \sigma^{(b)}\big)\Big),
\end{split}
\end{equation}
where we used the fact that, in the thermodynamic limit, $\lim_{N\to \infty}\frac1N f(\boldsymbol\xi_i)= \mathbb E f(\boldsymbol \xi)$. Then, we have the final form
\begin{equation}
\notag
\begin{split}
\mathbb E'Z_N^n&=\int d\mu\exp\Big(-\frac{\beta N}2 \sum_a (\tau_1^a)^2+\frac{\beta N}{\lambda}\sum_a\sqrt{1+\lambda X^{(a)}}-\frac{\beta N}2\sum_a K^{(a)}X^{(a)} -\frac{\alpha \beta N}{2}  \sum_{ab}p_{ab}q_{ab}\\&-\frac{\alpha N}{2}\log \text{det} (\mathbf 1 -\beta \mathbf q)+N \mathbb E \log \sum_{\boldsymbol \sigma}\exp\big(\beta\sum_a \sqrt{K^{(a)}} \xi^1 \tau_1^a \sigma^{(a)}+\frac{\alpha\beta}{2}\sum_{a,b}p_{ab}\sqrt{K^{(a)} K^{(b)}}\sigma^{(b)} \sigma^{(b)}\big)\Big).
\end{split}
\end{equation}
To go forward, we must assume the replica symmetry, which in our case turns out to be
\begin{eqnarray}
K^{(b)}&=&K,\\
X^{(a)}&=&X,\\
\tau_1^{(a)}&=&\sqrt{K}m,\\
q_{ab}&=&K\delta_{ab}+K q(1-\delta_{ab}),\\
p_{ab}&=&P\delta_{ab}+p(1-\delta_{ab}).
\end{eqnarray}
With this ansatz, we get
\begin{equation}
\notag
\begin{split}
\mathbb E'Z_N^n&=\int d\mu \exp\Big[-\frac{\beta Nn}2 K m^2+\frac{\beta Nn}{\lambda}\sqrt{1+\lambda X}-\frac{\beta Nn}2  KX -\frac{\alpha \beta N n}{2}  K p(1-q)\\&-\frac{\alpha N n}{2}\Big(\log[1-\beta K(1-q)]-\frac{\beta K q}{1-\beta K (1-q)}\Big)+nN \int Dy \log 2\cosh(\beta K m+y \sqrt{\alpha\beta K p})\Big].
\end{split}
\end{equation}
Then, we can apply the basic replica trick formula $A=\lim_{N\to \infty,n\to 0} (\mathbb E'Z_N^n-1)/nN$. To do this, we first compute the leading contribution of the partition function in the $N\to \infty$ limit by means of Laplace method, then we expand the exponential up to the order $\mathcal O(n)$. With simple algebra, we arrive to the result
\begin{equation}
\begin{split}
A_{\textrm{RS}}(\alpha,\beta,\lambda)&=-\frac{\beta }2 K m^2+\frac{\beta }{\lambda}\sqrt{1+\lambda X}-\frac{\beta KX }2 -\frac{\alpha \beta }{2}  K p(1-q)-\frac{\alpha }{2}\log[1-\beta K(1-q)]\\&+\frac{\alpha }{2}\frac{\beta K q}{1-\beta K (1-q)}+ \mathbb E_y\log 2\cosh(\beta K m+y \sqrt{\alpha\beta K p}),
\end{split}
\end{equation}
which is precisely the quenched pressure reported in \eqref{main-res}.

\end{document}